\documentclass[12pt]{article}
\usepackage[onehalfspacing]{setspace}
\usepackage{amsfonts}
\usepackage{amsmath}
\usepackage{amssymb}
\usepackage{mathtools}
\usepackage{amsthm}
\usepackage{comment}
\usepackage{xfrac}
\usepackage{graphicx}
\usepackage{libertine,libertinust1math}
\usepackage{accents}
\usepackage[dvipsnames]{xcolor}
\usepackage{pdfpages}
\usepackage{enumitem}
\usepackage{tikz-cd}
\usepackage{floatrow}
\usepackage{pgfplots}
 \usepgfplotslibrary{fillbetween}
\pgfplotsset{compat=1.18}
\usepackage[round,longnamesfirst]{natbib} 
\usepackage{multirow}
\usepackage{dutchcal}
\usepackage[margin=1.25in]{geometry}
\usepackage{caption}
\usepackage{subcaption}
\usepackage{epigraph}
\usetikzlibrary{calc}
\usetikzlibrary{arrows.meta}
\usepackage{hyperref}
\usepackage[noabbrev,capitalise,nameinlink]{cleveref}
\usepackage{chngcntr}
\usepackage{apptools}
\usepackage[title]{appendix}
\crefname{appendix}{appendix}{appendices}
\Crefname{appendix}{Appendix}{Appendices}
\Crefname{claim}{Claim}{Claims}

\definecolor{DarkBlue}{rgb}{0.0, 0.0, 0.55}
\definecolor{DarkRed}{rgb}{0.55,0.0,0.0}
\hypersetup{
pdftitle={Uncharted Waters: Selling a New Product Robustly},
pdfauthor={Kun Zhang},
}
\hypersetup{
     colorlinks = true,
     linkcolor = DarkRed,
     citecolor = DarkBlue,
     urlcolor = Fuchsia,
}

\newcommand{\ubar}[1]{\underaccent{\bar}{#1}}

\DeclareMathOperator*{\argmax}{arg\,max}

\newtheorem{theorem}{Theorem}
\newtheorem{lemma}{Lemma}
\newtheorem{proposition}{Proposition}
\newtheorem{claim}{Claim}

\newtheorem{corollary}{Corollary}

\theoremstyle{definition}

\definecolor{PolicyUniform}{rgb}{0.33,0.53,0.74}
\definecolor{PolicyFull}{rgb}{0.50,0.39,0.70}
\definecolor{PolicyMixture}{rgb}{0.70,0.30,0.30}
\definecolor{myblue}{RGB}{142,170,204}
\definecolor{mymaroon}{RGB}{214,157,157}
\definecolor{mygreen}{RGB}{151,188,151}
\definecolor{mypurple}{RGB}{166,156,204}
\definecolor{curvegray}{RGB}{85,85,85}

\begin{document}
\begin{titlepage}
\title{Uncharted Waters: Selling a New Product Robustly}
\author{Kun Zhang\thanks{School of Economics, University of Queensland. Email: \href{mailto:kun@kunzhang.org}{kun@kunzhang.org}. 
\newline 
I am indebted to Hector Chade for his continued guidance and encouragement, and I am grateful to Heski Bar-Isaac, Andreas Kleiner, and Mark Whitmeyer for insightful discussions and suggestions. I also thank numerous colleagues and seminar and conference audiences for invaluable feedback, especially the discussants Jesse Bull, Zheng Gong, Nicolas Pastrian, and Luca Sandrini. All errors are my own.
}}
\date{\today}

\maketitle
\thispagestyle{empty} 

\begin{abstract}
    New products often involve uncertainty about product fit, while sellers may also be unsure about what alternatives buyers face. I study a seller of a new product who sets a price and provides product-fit information before the buyer decides whether to search for an outside option. The buyer knows the outside-option distribution; the seller knows only its mean and bounds and maximizes guaranteed profit across compatible distributions. Information provision has two roles: it can hedge against uncertainty about the outside-option distribution and deter search by making buying without search attractive. Price governs their relative value: a higher price raises revenue per sale but makes search deterrence harder. The results help explain why new products differ in product-fit information and show that lower search costs can raise prices and make information noisier, challenging the presumption that easier search necessarily benefits consumers.
\end{abstract}

\end{titlepage}
\setcounter{page}{1}

\section{Introduction}
Rapid technological progress has brought a growing number of new products to market. Examples include e-readers, electric vehicles, smartwatches, and virtual-reality headsets. Due to their novelty, sellers often have significant control over what buyers can learn about them. In particular, sellers can provide information by offering free trials, product samples, and demonstrations. Buyers can nevertheless compare the prices and features of related existing products. Yet sellers may know little about which alternatives buyers include in these comparisons, how buyers learn about them, or how attractive those alternatives are. This uncertainty gives sellers reason to seek robust selling strategies that perform reasonably well across the different alternatives buyers may face.

In this paper, I examine the following questions: How do robustness concerns shape the seller's optimal pricing and information provision strategies? Would the buyer be better off if learning about her alternatives became easier? Finally, what insights do the answers to these questions offer regarding the sale of different types of new products?

To address these questions, I study a model in which a seller faces a buyer whose match value with the product is either high or low. Although neither party knows the match value, the prior probability of a high match value is common knowledge. The seller posts a price and chooses how much information to provide about the match value.
The buyer also has access to an outside option representing her best alternative to the seller's new product. She understands the distribution of this alternative, although its realized value is initially unknown. After observing the price and receiving the information, the buyer updates her beliefs about the match value. She then decides whether to incur a cost to discover the value of the outside option --- a process I term ``search'' --- or to buy the seller's product directly. If she searches, she may still return and buy from the seller if the outside option proves less attractive than the seller's offer.

The seller is only partially informed about the buyer's outside-option distribution, knowing only its mean and bounds. 
This assumption formalizes the uncertainty described above: in new-product markets, the buyer may be able to evaluate her own search prospects, whereas the seller has only coarse information about them.
The mean and bounds capture the seller's broad knowledge of the average attractiveness and feasible range of existing alternatives. 
They leave unspecified other payoff-relevant features of the distribution that are difficult to infer from data, especially in such markets, where demand histories are short and the relevant comparison set remains uncertain. 
I therefore model the seller as seeking robustness against this remaining distributional uncertainty: she chooses a price and an information provision policy to maximize the profit she can guarantee across all outside-option distributions consistent with her information. This avoids evaluating her strategy by its expected profit under a precise probabilistic model of the buyer's search environment, which she may be unable to substantiate.

Because the buyer initially does not know whether the product is a good match, the information provided by the seller shapes the product impression she forms before deciding whether to search. This ability to shape impressions gives information two roles. 
First, it can hedge against uncertainty about the buyer's alternatives. When the information provided spreads the buyer's product impressions continuously and evenly, the probability of sale at a fixed price depends on her alternatives only through their known average attractiveness. This hedging role makes demand robust to the unknown shape of the buyer's outside-option distribution. 
Second, information can deter search. Because search is costly, a sufficiently favorable impression can make the buyer purchase without inspecting the outside option. Under the seller's uncertainty, such an impression creates ``safe'' demand only when it induces purchase without search for every outside-option distribution consistent with her information. The outside option then cannot draw the buyer away, regardless of how attractive it would have been.
Price determines how the two roles of information translate into profit. A higher price raises the profit margin from each sale but makes safe demand harder to create by reducing the buyer's payoff from buying immediately. It also makes return after search less likely by narrowing the range of outside-option values for which the buyer would still prefer the seller's product.

When search costs are low, inspecting the outside option is relatively attractive, so creating safe demand after favorable impressions would require the seller to set too low a price. 
The seller instead charges a higher price and relies on the hedging role of information: the optimal information spreads product impressions continuously and evenly, stabilizing demand across possible outside-option distributions.
The information provided is noisy: the buyer's impression is typically neither strongly favorable nor clearly unfavorable.

As the search cost increases, a favorable impression can support purchase without search at a higher price because inspecting the outside option becomes less attractive. When the search cost is sufficiently high, providing full information --- revealing whether the match value is high or low --- becomes optimal. 
If the product is revealed to be a good match, the buyer's willingness to pay is the highest possible. Because search is costly, this willingness to pay can support safe demand without requiring the seller to set too low a price.
The gain from generating safe demand whenever the match value is high then outweighs the value of generating less decisive impressions under which the buyer might search and return.

For intermediate search costs, the seller's strategy also depends on the prior probability of a high match value --- that is, on how promising the product looks before product-fit information is provided.
If the product does not look especially promising, the seller has limited scope to generate very favorable impressions, because the product is unlikely to fit the buyer especially well. 
It is then not worthwhile to lower the price to turn such impressions into safe demand, so the seller uses the same form of information as in the low-search-cost case, spreading product impressions continuously and evenly. If the product looks more promising, safe demand becomes more attractive, and the seller may want some impressions to be favorable enough to induce purchase without search. When the buyer's alternatives are relatively unattractive on average, return-after-search demand also remains valuable. The seller then combines elements of the two forms of information just described: as under full information, she sometimes gives the buyer a clear indication that the product is a strong fit, creating safe demand; as in the low-search-cost case, she otherwise gives the buyer less decisive impressions that keep the product attractive enough for her to return after search. When the buyer's alternatives are relatively attractive on average, return-after-search demand is less reliable, and full information is optimal.

The model also yields notable comparative statics. A higher search cost makes inspecting the outside option less attractive, so one might expect the seller to charge a higher price. This force is present, but the price can nevertheless jump downward. 
As discussed above, when search costs are low, the seller does not pursue safe demand and instead sets a relatively high price.
As search costs rise, the price cut needed to create safe demand becomes smaller.
Eventually, it becomes small enough that creating safe demand is worthwhile, so the seller switches to creating safe demand and the price jumps down.
Product-fit information also generally becomes more informative as search costs rise: lower search costs lead the seller to provide less decisive product-fit information, while higher search costs make sharper product-fit information more attractive. Thus, technological advances that reduce search costs need not benefit consumers: for some new products, such advances can lead to higher prices and less informative product-fit information.

The results also offer an explanation for why different kinds of new products warrant different information provision policies. Some products are \emph{revolutionary}, such as the iPhone and the 3D printer. Others are \emph{evolutionary} --- incremental improvements on existing products --- such as smart thermostats and energy-saving light bulbs. Still others are ``trade-off products'': they offer substantial benefits along new dimensions while sacrificing some familiar functionality. Portable wireless speakers, for instance, provide greater convenience but at the expense of sound quality. One interpretation of search cost is that it captures the difficulty of identifying and evaluating the best available alternative. Evolutionary products are relatively easy to compare with existing alternatives because they differ only marginally, whereas products involving less familiar trade-offs are harder to compare.

Under this interpretation, evolutionary products favor less decisive information provision, such as short trial periods or free trials with limited features. In contrast, trade-off products favor more decisive information provision, which distinguishes consumers for whom the product is a good fit from those for whom it is not and allows the seller to target the former. Policies such as extended trial periods and comprehensive demonstrations can serve this role. For revolutionary products that already look highly promising before detailed product-fit information is provided, it is attractive to give some consumers a clear indication that the product is a strong fit while leaving others with less decisive impressions. These patterns may help explain the variety of information provision policies observed among new products.

The implications discussed above rely on two key features of the model: search frictions and the seller's robustness concerns. Without search frictions, the buyer always inspects the outside option, and therefore favorable impressions cannot create safe demand. The seller still faces uncertainty about the buyer's outside-option distribution, and the hedging motive makes it optimal to spread product impressions continuously and evenly. If, instead, the seller knows the full distribution of the buyer's outside option, the hedging role disappears. Full information is then always optimal: search frictions still affect the price, but not the form of information provision. Thus, neither benchmark alone delivers the joint price--information patterns that are central to the implications for different kinds of new products.

The remainder of the paper is organized as follows. The rest of this section discusses related literature. \autoref{s:model} introduces the model. \autoref{s:main} presents the main results of the paper. \autoref{s:discussion} concludes.

\subsection{Related Literature}
The literature on selling new products has largely focused on strategic pricing,\footnote{For a survey, see \cite{c09r}. Many papers cited therein study the pricing dynamics of new products, an issue from which I abstract.} while a smaller strand studies settings in which sellers also decide what information to provide about their products.\footnote{In \cite{mr86}, the seller of a new product chooses both a price and an advertising expenditure level. Although the latter serves as a signal of product quality, it does not convey any direct information about the product.} 
\cite{hm96} study how demonstration length affects purchase probabilities for various types of new products.
\cite{flo} study how information provision and consumer reviews interact across product generations.
\cite{salamat2025trial} characterizes when trials providing product-fit information allow the seller to extract full surplus.
\cite{bcg} investigate competition between a firm introducing a new product of unknown match value and a firm offering an established product with a known match value. They find that partial information provision is optimal when the innovative firm sets its price before choosing what information to provide, whereas full information provision is optimal when it chooses what information to provide before setting its price. My paper adds to this strand by studying how search frictions and robustness concerns shape a new-product seller's joint choice of price and information provision.

In my model, the seller can use information provision to shape the buyer's search decision. This links my work to two strands of the literature on consumer search deterrence. One studies price-based search-deterrence tactics. \cite{az16} study how sellers can use pricing tools, such as buy-now discounts, exploding offers, and nonrefundable deposits, to deter buyers from exploring competing products.\footnote{The theoretical predictions in \cite{az16} are tested experimentally in \cite{bvw} and \cite{pan2023commitment}. \cite{liu2025personalized} study a variant of the model in \cite{az16} in which the seller knows the buyer's outside-option value.} The other studies search obfuscation, whereby sellers deliberately raise search costs.\footnote{For a review of this literature, see \cite{ellison2016price}.} \cite{bicc} consider a monopoly seller who can use marketing strategies to make it more costly for the buyer to assess how well the seller's product suits her. \cite{ellison2012search} extend the model of \cite{stahl1989oligopolistic} by allowing the buyer's incremental search cost to rise with her cumulative search effort and by enabling sellers to take actions that make onward search more costly.

My setup builds on the monopoly model of \cite{az16} but highlights a different search-deterrence channel: information provision. \cite{wang2017advertising}, \cite{feng2024monopoly}, \cite{koessler2024sources}, and \cite{pham2025discriminatory} study related information-based channels. In these papers, as in \cite{bicc}, the buyer searches for information about the seller's own product. Unlike in \cite{bicc}, however, the search cost remains fixed, and the seller deters search by providing information that reduces the value of learning more about the product. In my model, by contrast, the buyer's search is over outside options. Information deters search by making immediate purchase attractive enough for the buyer to forgo inspecting those alternatives.

At a higher level, this paper lies at the intersection of two literatures: robust monopoly pricing and monopoly pricing with information provision and consumer search. The latter includes \cite{ar06} and \cite{lyu2023information}, in which purchase requires search and search deterrence therefore does not arise, as well as the information-based search-deterrence papers discussed above. Existing work on robust monopoly pricing typically studies uncertainty about the buyer's value for the seller's own product.\footnote{Representative contributions in this literature include \cite{bs08,bs11}, \cite{carrasco2018optimal}, \cite{d18s}, \cite{hk20}, and \cite{che2025robustly}. For a recent survey of robust contracting, see \cite{carroll2019robustness}.} Here, the seller is instead uncertain about the buyer's outside-option distribution. Because the seller also chooses what information to provide, the robustness concern shapes the joint design of price and information provision.

The seller's robust selling problem can be viewed as a contest in which she chooses a price and an information provision policy while an adversarial decision maker selects the outside-option distribution to minimize the seller's profit. This connects my work, from a technical perspective, to the literature on competitive information provision, which traces back to \cite{os10}.
In the contest models studied by \cite{bc15,bc18} and \cite{ak20,ak21}, multiple senders each commit to an information provision policy and compete to be selected by a receiver. \cite{hkb} study competing sellers who choose both prices and information provision policies. \cite{au2023attraction,au2024attraction}, \cite{he2023competitive}, \cite{boleslavsky2025limits}, and \cite{hwang2025competitive} add search frictions to information-design contests. Here, however, the contest is asymmetric in two respects: even at a fixed price, the seller and the adversarial decision maker face different best-response problems; moreover, only the seller chooses the price, which affects how the buyer compares the seller's product with the outside option. The interaction between these two sources of asymmetry requires new techniques to solve the problem. Economically, the paper focuses on the seller's robust pricing and information provision decisions rather than on competition among senders.

Finally, this work is related to papers on information design in which the sender is uncertain about aspects of the receiver's decision problem and maximizes her worst-case payoff.\footnote{Two other papers consider a different objective: the sender minimizes the worst-case gap between her payoff and what she could obtain if the uncertain environment were known. In \cite{btxz}, the uncertainty concerns the receiver's utility function. \cite{parakhonyak2026persuasion} assume that the sender is uncertain about the probability distribution over the states of the world, the receiver's utility function, and the receiver's prior.} In \cite{kosterina2022persuasion}, the sender is uncertain about the receiver's prior, which may depart from a reference prior to some degree. \cite{hw21} and \cite{dp22} assume that the sender and receiver share a common prior over the state, while the uncertainty concerns the receiver's additional information. Most closely related is \cite{sapiro2024persuasion}, who considers a setting in which the receiver takes the sender's preferred action if and only if the posterior mean of the state exceeds the receiver's outside option, and the sender is uncertain about the outside-option distribution. 
Adding the pricing channel changes the objective from maximizing the probability of the sender's preferred action to maximizing guaranteed profit, generating new implications for information provision.

\section{The Model} \label{s:model}
Buyer's match value with the product is \(x\in\{0,1\}\), where \(x=1\) denotes the high match value and \(x=0\) the low match value. The common prior assigns probability \(\mu \coloneqq \mathbb{P}(x=1) \in (0,1)\) to the high match value. Since \(x\) is binary, \(\mu\) is also the \emph{prior mean match value}. Initially, neither Seller nor Buyer knows the realization of \(x\). Seller's production cost is normalized to zero, and therefore revenue and profit coincide. 
Seller chooses a price \(p\) and an \textbf{information provision policy} \((\sigma,Y)\), consisting of a signal space \(Y\) and a mapping \(\sigma:\{0,1\}\to\Delta(Y)\). After observing a signal realization \(y\in Y\), Buyer forms the posterior belief \(w=\mathbb P(x=1\mid y)\in[0,1]\), which is also her expected match value from Seller's product. This posterior belief formalizes what the introduction refers to as Buyer's impression of the new product.

For any \(m\in[0,1]\), let \(\mathcal M(m)\) denote the set of distributions on \([0,1]\) with mean \(m\). It is well known that a cumulative distribution function (CDF) \(H\) over posteriors can be induced by an information provision policy if and only if \(H \in \mathcal{M}(\mu)\), that is, \(\mathbb{E}_H[w] = \mu\);
in words, the expected posterior equals the prior.\footnote{See, for example, \cite{kg11}.} The analysis can therefore be recast as Seller choosing \((p,H)\in[0,1]\times\mathcal M(\mu)\) instead of \((p,(\sigma,Y))\). I call \((p,H)\) a \textbf{selling strategy}, and refer to the choice of \(H\) as the choice of an information provision policy. If the realized posterior is \(w\), Buyer's \textbf{net value} from purchasing Seller's product is \(w-p\).

Buyer has unit demand and an outside option whose value \(v\) is initially unknown. Buyer knows that \(v\) is distributed according to a CDF \(G\) on \([0,1]\), but must incur a search cost \(s > 0\) to discover its realization. One way to interpret this is as a reduced form of Buyer's sequential search. I assume that \(v\) is independent of the match value \(x\) and of Seller's signal. Seller cannot observe \(v\) and does not know \(G\); she knows only that \(G\) is a distribution on \([0,1]\) with mean \(\xi\in(0,1)\). To focus on the interesting cases in which Buyer prefers searching to buying nothing, I assume \(s<\xi\).

I study Seller's problem of maximizing the revenue guarantee, defined as the infimum of expected revenue over all outside-option distributions on \([0,1]\) with mean \(\xi\). Metaphorically, after Seller chooses \((p,H)\), an adversarial decision maker, henceforth Nature, seeks to minimize Seller's revenue by choosing a distribution \(G\in\mathcal M(\xi)\). This interpretation will be used below, as it is helpful in solving Seller's optimization problem and interpreting the results.

Let
\[S_G(t) \coloneqq \mathbb{E}_{G}[\max \{t, v\}]-t=\int_{t}^{1}(v-t) \, \mathrm{d} G(v)\]
denote the expected benefit of search when Buyer's net value of purchasing Seller's product is \(t\) and the outside option is distributed according to some \(G\). Let \(a\) be the unique solution to\footnote{Throughout, the dependence of \(a\) on \(G\) and \(s\) is suppressed. Existence follows because \(S_G\) is continuous, \(S_G(0)=\xi>s\), and \(S_G(1)=0<s\). Uniqueness follows because \(s>0\) and \(S_G\) is strictly decreasing whenever it is positive. Unless explicitly specified, I use ``increasing'' and ``decreasing'' in the weak sense; that is, ``increasing'' means ``weakly increasing.'' The qualifier ``strictly'' will be added when needed.}
\begin{equation} \label{eq:res-value}
S_G(a)=s.
\end{equation}
Buyer purchases Seller's product without search whenever the expected benefit from search is no more than the search cost, that is, \(S_G(w-p)\le s\). Since \(S_G\) is decreasing, \eqref{eq:res-value} implies that this is equivalent to \(w-p\ge a\). Intuitively, \(a\) is the minimum net value at which Buyer purchases without search, or, in jargon, the \emph{reservation value} of the outside option.
It can be checked that \(a \in [\xi-s, 1-s/\xi]\). The lower bound is attained by \(\delta_{\xi}\), the degenerate distribution at \(v=\xi\), while the upper bound is attained uniquely by the binary distribution that assigns probability \(1-\xi\) to \(0\) and probability \(\xi\) to \(1\).\footnote{Let \(G_B\) denote this binary distribution. For \(a\in[0,1]\),
\[
S_{G_B}(a)=\mathbb E_{G_B}[\max\{v,a\}]-a
=\xi+(1-\xi)a-a=\xi(1-a).
\]
Hence \eqref{eq:res-value} gives \(a=1-s/\xi\). The intuition is standard in search theory \citep[e.g.,][]{ks74}: holding the mean fixed, a more dispersed outside-option distribution raises the value of search because very good options help Buyer, while very bad options are less harmful since Buyer can always return to buy Seller's product. Thus, among distributions on \([0,1]\) with mean \(\xi\), the most concentrated distribution \(\delta_\xi\) yields the lowest reservation value, while the most dispersed binary distribution \(G_B\) yields the highest.} By taking convex combinations of these two distributions, any \(a\in(\xi-s,1-s/\xi)\) can be achieved.

If Buyer instead searches, that is, if \(w-p<a\), she pays the search cost and observes the realization \(v\) of the outside option. She then returns to buy Seller's product if and only if the outside option turns out to be worse than Seller's offer, that is, if \(w-p>v\).\footnote{The tie-breaking convention is embedded in the two inequalities above: Buyer does not search when indifferent between searching and not searching, and she does not return to Seller when indifferent between Seller's offer and her outside option. Under this convention, a minimizing outside-option distribution need not exist for a given selling strategy, although Nature can choose a sequence of distributions in \(\mathcal M(\xi)\) under which Seller's expected revenue converges to the revenue guarantee. The analysis below nevertheless shows that Seller's maximization problem admits a solution. \label{fn:tie-breaking}} I assume that Seller cannot recognize whether Buyer is a first-time visitor; consequently, she cannot adjust the price when Buyer comes back from search.\footnote{For the same reason, Seller cannot benefit from offering a menu of prices depending on Buyer's report of the realized outside option: no matter whether Buyer has searched, she would make whichever report yields the lowest price. This assumption serves two purposes: (a) to focus on how search frictions and Seller's robustness concerns influence the initial ``price tag'' at the launch of a new product, and (b) to examine how information provision can help deter search without resorting to price-based search deterrence tactics. \href{https://kunzhang.org/file/roboa.pdf}{Supplementary Appendix} C.1 details what happens if Buyer's identity is recognizable to Seller, enabling Seller to make exploding offers or engage in intertemporal price discrimination.}

The timing of the game is as follows:
\begin{enumerate}[noitemsep,topsep=0pt]
    \item Seller chooses a selling strategy \((p,H)\);
    \item Nature observes Seller's choice and chooses a distribution \(G\).
       \item Buyer observes the price \(p\), the distribution \(G\), and a posterior \(w\) drawn according to \(H\). If \(w-p\ge a\), she buys immediately. Otherwise, she pays the search cost \(s\) and observes a realization \(v\) from \(G\).
    \item After search, Buyer returns to Seller and buys if and only if \(w-p>v\).
\end{enumerate}

Given a selling strategy \((p,H)\) and an outside-option distribution \(G\), Seller's expected revenue can be written compactly as\footnote{Here \(1-H(p+\min\{a,v\})\) is interpreted according to the tie-breaking rule above: when \(v<a\), the mass at posterior \(w=p+v\) is excluded from demand; when \(v\geq a\), the mass at posterior \(w=p+a\) is included in demand.}
\begin{equation} \label{eq:exp-revenue}
\Pi(p,H \mid G) \coloneqq p \,\mathbb{E}_G[1-H(p+\min\{a,v\})].
\end{equation} 
This is because conditional on the outside-option value \(v\), Buyer eventually purchases from Seller if \(w>p+v\) when \(v<a\), and if \(w\ge p+a\) when \(v\ge a\). Thus, \(\mathbb{E}_G[1-H(p+\min\{a,v\})]\) is the probability of eventual purchase, or the demand that Seller faces, under \(G\).

Seller's revenue-guarantee maximization problem can therefore be written as
\[
\max_{(p,H)\in[0,1]\times\mathcal M(\mu)}
\inf_{G\in\mathcal M(\xi)}
\Pi(p,H\mid G).
\]
A solution to this problem is called a \textbf{robustly optimal selling strategy}; its components are called a \textbf{robust price} and a \textbf{robust information provision policy}, respectively.
Equivalently, Seller solves
\[
\max_{p\in[0,1]}
\left\{
\max_{H\in\mathcal M(\mu)}
\inf_{G\in\mathcal M(\xi)}
\Pi(p,H\mid G)
\right\}.
\]
Let \(\Phi(p) \coloneqq \max_{H \in \mathcal{M}(\mu)} \inf_{G \in \mathcal{M}(\xi)} \Pi(p, H \mid G)\).
Seller's problem can be solved in two steps: first, for a fixed price \(p\), Seller chooses \(H \in \mathcal{M}(\mu)\) to maximize \(\inf_{G \in \mathcal{M}(\xi)} \Pi(p,H \mid G)\), namely Seller's revenue guarantee at that price; second, Seller chooses \(p \in [0,1]\) to maximize \(\Phi(p)\). 

\subsection{Discussion of the Model}
The assumption that Buyer knows the outside-option distribution \(G\), while Seller knows only its mean and bounds, captures a distinction between Buyer's individual search problem and the coarse features Seller can credibly infer about it. Buyer may know \(G\) because her comparison set, search process, and access to alternatives determine which outside options she expects to encounter. Seller, especially for a new product, may have enough information from surveys, related products, or aggregate market evidence to discipline the average attractiveness and feasible range of those alternatives. Such information, however, need not identify the full shape of \(G\), which is nonetheless payoff-relevant because it affects Buyer's incentive to search. Specifying a prior over possible \(G\)'s would require Seller to assign probabilities to precisely these unobserved distributional features, which can be overly demanding. I therefore formulate Seller's problem as a robust selling problem: Seller chooses a price and an information provision policy to maximize the revenue she can guarantee given the information she can credibly use. Such a guarantee is especially relevant for new-product launches, where early pricing and information choices can shape adoption and subsequent market performance, yet must often be made before Seller has enough data to learn a reliable probabilistic model of Buyer's search environment.

I assume that Buyer's match value with Seller's new product, \(x\), takes one of two values, \(0\) or \(1\). One can interpret \(x=1\) as Buyer liking the product and \(x=0\) as not liking it. Although stark, this binary specification is crucial for tractability: a distribution over posteriors \(H\) is feasible if and only if its mean equals the prior probability of liking, \(\mu\). With more than two match values, however, satisfying the mean constraint would no longer be sufficient for \(H\) to be feasible; feasibility would depend on the full prior distribution of \(x\). Even at a fixed price, Seller's robust information provision problem subject to these prior-dependent restrictions can be difficult to characterize; optimizing over price requires solving such a problem for every price. \href{https://kunzhang.org/file/roboa.pdf}{Supplementary Appendix} C.2 further discusses this tractability problem for a continuously distributed match value and shows that some key qualitative insights continue to hold.

\href{https://kunzhang.org/file/roboa.pdf}{Supplementary Appendix} C.2 also discusses several other modeling assumptions. Allowing Seller to randomize over prices or allowing Buyer to have access to a ``safe'' outside option that can be consumed without search does not qualitatively change the main takeaways. The appendix also discusses alternative assumptions about Seller's knowledge of the search cost and the outside-option distribution. I keep the baseline assumptions in the main text as they isolate the roles of price, information, and search in the clearest form.

\section{Main Results} \label{s:main}
This section presents the main results of the paper. \autoref{ss:two-roles} explains the two roles of information provision: it can hedge against uncertainty about Buyer's outside-option distribution and deter search by making buying without search attractive. 
It also explains how price shapes the trade-off between these roles: a higher price raises revenue per sale but makes search deterrence harder.
Building on this logic, \autoref{ss:robust-optimal} characterizes the robustly optimal selling strategy.
\autoref{ss:cs} studies comparative statics. All proofs are collected in \Cref{prf}.

\subsection{Two Roles of Information Provision}\label{ss:two-roles}

The first role of information provision is to hedge demand against Nature's choice of the outside-option distribution \(G\). To see the logic, fix a price \(p\). Although Nature's choice of \(G\) affects Seller's demand \(\mathbb{E}_G[1-H(p+\min\{a,v\})]\) through the distribution of \(\min\{a,v\}\), its mean is pinned down by what Seller knows: using \eqref{eq:res-value},
\[\mathbb{E}_G[\min\{a,v\}]=\mathbb{E}_G[v+a-\max\{a,v\}]=\xi-s.\]
That is, Nature can change how \(\min\{a,v\}\) is distributed around its mean, but not the mean itself.

This observation makes affine segments in the distribution over posteriors \(H\) valuable.  Suppose \(H\) is affine on the relevant values of \(p+\min\{a,v\}\). Then
\[\mathbb{E}_G[1-H(p+\min\{a,v\})]
=1-H\bigl(p+\mathbb{E}_G[\min\{a,v\}]\bigr)
=1-H(p+\xi-s);\]
that is, the affine segment makes demand depend only on the known number \(\xi-s\), rather than on the full distribution of outside options. This prevents any particular choice of the outside-option distribution from giving Nature a significant advantage, yielding the desired robustness. 

The second role of information is to deter search by generating sufficiently favorable posteriors. For a fixed outside-option distribution \(G\), posterior \(w\) induces Buyer to purchase without search whenever \(w-p\ge a\).
When Nature may instead choose any \(G\in\mathcal M(\xi)\), posterior \(w\) creates \emph{safe demand} only if Buyer purchases without search for every such distribution.
The largest reservation value Nature can induce is \(1-s/\xi\). Therefore, posterior \(w\) creates safe demand if and only if
\[
w-p\ge 1-s/\xi,
\]
or equivalently,
\[
w \ge p+1-s/\xi.
\]
If this condition fails, Nature can choose a distribution for which the reservation value strictly exceeds \(w-p\), inducing Buyer to search. Since posteriors cannot exceed one, safe demand can be created only when \(p\le s/\xi\).

To summarize, if \(p>s/\xi\), even the most favorable posterior cannot create safe demand. At such prices, the hedging logic described above calls for affine segments in the distribution over posteriors. If \(p\le s/\xi\), by contrast, posteriors in \([p+1-s/\xi,1]\) create safe demand: Buyer buys without search for every possible outside-option distribution. Information can then also serve a search-deterrence role by placing probability on such posteriors.

\subsection{Robustly Optimal Selling Strategy}\label{ss:robust-optimal}
The next result shows how the two roles of information provision, together with the way price shapes the trade-off between them, determine the robustly optimal selling strategy. To state the result, I first introduce some notation. Let \(\ubar{\mu}\) and \(\bar{\mu}\) be two thresholds for the prior \(\mu\) with \(0 < \ubar{\mu} < \bar{\mu} < 1\) for all \(\xi \in (0,1)\) and \(s \in (0,\xi)\) (see \Cref{ss:omitted-details} for formal definitions). Let \(p^*\) and \(p^{**}\) be given by
\begin{equation} \label{eq:no-det-price}
p^*\coloneqq\begin{cases} \frac{1-\sqrt{2(\xi-s)-(\xi-s)^{2}}}{1-\xi+s} & \text{if } \mu \le \ubar{\mu} \\ 2\mu-1 & \text{if } \ubar{\mu} < \mu \le \bar{\mu}  \\  1-\sqrt{\xi-s} & \text{if } \mu > \bar{\mu} \end{cases}
\end{equation}
and \(p^{**} \coloneqq s/\xi\), respectively.

Say that Seller uses \textbf{uniform information} if the induced distribution over posteriors \(H\) has no mass point at \(1\) and is uniform on a single interval ending at \(1\), apart from a possible mass point at \(0\).\footnote{When the prior \(\mu\) is small, the uniform interval alone would have mean above \(\mu\). The expected posterior is then brought back to \(\mu\) by assigning the remaining probability to posterior \(0\). Using the lowest posterior minimizes the probability needed for this adjustment and leaves more probability on the uniform segment.} This is illustrated in \autoref{subfig:unif}. 
The information provided about the match value is ``noisy'' under uniform information: each posterior in the interior of the uniform interval can arise under both the high and low match values. Say that Seller uses \textbf{full information} if \(H\) is the binary distribution on \(\{0,1\}\).  This distribution, illustrated in \autoref{subfig:full}, fully reveals the match value: if the realized posterior is 1 (0), the match value must be high (low).  Say that Seller uses \textbf{mixture information} if \(H\) has a mass point at \(1\) and is uniform on a single interval ending at \(1\). As shown in \autoref{subfig:mix}, mixture information, like uniform information, is noisy; however, when the realized posterior is \(1\), the match value must be high. The formal definitions of these information provision policies can be found in \Cref{ss:omitted-details}.

\begin{figure}
\centering
\begin{subfigure}[b]{0.31\textwidth}
\centering
\begin{tikzpicture}[scale=0.62,
    >={Stealth[length=2mm,width=1.5mm]}]
    \draw[<->,thick] (0,4.5) node[left]{\small \(H\)} -- (0,0) -- (5.3,0) node[below]{\small \(w\)};
    \draw[dashed] (4.4,0) node[below]{\small \(1\)} -- (4.4,4.0) -- (0,4.0) node[left]{\small \(1\)};
    \draw[ultra thick,PolicyUniform] (0,0.7) -- (1.35,0.7) -- (4.4,4.0) -- (5.0,4.0);
    \node[below] at (1.35,0) {\small \(\ell\)};
\end{tikzpicture} 
\caption{Uniform information}
\label{subfig:unif}
\end{subfigure}\hfill
\begin{subfigure}[b]{0.31\textwidth}
\centering
\begin{tikzpicture}[scale=0.62,>={Stealth[length=2mm,width=1.5mm]}]
    \draw[<->,thick] (0,4.5) node[left]{\small \(H\)} -- (0,0) -- (5.3,0) node[below]{\small \(w\)};
    \draw[dashed] (4.4,0) node[below]{\small \(1\)} -- (4.4,4.0) -- (0,4.0) node[left]{\small \(1\)};
    \draw[ultra thick,PolicyFull] (0,2.0) -- (4.4,2.0);
    \draw[ultra thick,PolicyFull] (4.4,4.0) -- (5.0,4.0);
\end{tikzpicture}
\caption{Full information}
\label{subfig:full}
\end{subfigure}\hfill
\begin{subfigure}[b]{0.31\textwidth}
\centering
\begin{tikzpicture}[scale=0.62,>={Stealth[length=2mm,width=1.5mm]}]
    \draw[<->,thick] (0,4.5) node[left]{\small \(H\)} -- (0,0) -- (5.3,0) node[below]{\small \(w\)};
    \draw[dashed] (4.4,0) node[below]{\small \(1\)} -- (4.4,4.0) -- (0,4.0) node[left]{\small \(1\)};
    \draw[ultra thick,PolicyMixture] (0,0) -- (1.1,0) node[below,black]{\small \(p^{**}\)} -- (4.4,3.0);
    \draw[ultra thick,PolicyMixture] (4.4,4.0) -- (5.0,4.0);
\end{tikzpicture}
\caption{Mixture information}
\label{subfig:mix}
\end{subfigure}
\caption{The three information policies that can be robustly optimal.} 
\floatfoot{The uniform information panel depicts the case with a mass point at 0; when that mass point is absent, the same affine segment starts from \(H(\ell)=0\).  Here \(\ell=p^*\) when \(\mu\le\bar\mu\), while \(\ell=2\mu-1\) when \(\mu>\bar\mu\). Unlike uniform information, both full and mixture information place positive mass on posterior \(1\).}
\label{fig:policies}
\end{figure}

The characterization below is stated in terms of three search-cost cutoffs. Let \(B_1(\xi) \coloneqq \xi(\xi-1)^2/(\xi^2+1)\), \(B_2(\xi) \coloneqq \xi(\xi-1)^2/(\xi+1)^2\), and \(B_3(\xi) \coloneqq \xi - 2\xi^2\). It can be checked that \(B_1(\xi) > B_2(\xi)\) and \(B_1(\xi) > B_3(\xi)\) for all \(\xi \in (0,1)\).

\begin{theorem}[Robustly optimal selling strategy]\label{theorem:optimal}
If \(s\ge B_1(\xi)\), then full information is optimal, and the robust price is \(p^{**}=s/\xi\). If \(s\le B_2(\xi)\), then uniform information is optimal, and the robust price is \(p^*>s/\xi\). If \(B_2(\xi)<s<B_1(\xi)\), there are two cases:
\begin{itemize}[leftmargin=2em]
    \item[(1)] If \(B_3(\xi)\le s<B_1(\xi)\), then there exists \(\hat{\mu}\in(0,1)\) such that for \(\mu<\hat{\mu}\), uniform information is optimal, and the robust price is \(p^*>s/\xi\); and for \(\mu\ge\hat{\mu}\), full information is optimal, and the robust price is \(p^{**}=s/\xi\).
    \item[(2)] If \(B_2(\xi)<s<B_3(\xi)\), then there exists \(\check{\mu}\in(0,1)\) such that for \(\mu<\check{\mu}\), uniform information is optimal, and the robust price is \(p^*>s/\xi\); and for \(\mu\ge\check{\mu}\), mixture information is optimal, and the robust price is \(p^{**}=s/\xi\).
\end{itemize}
\end{theorem}

Despite its case distinctions, \autoref{theorem:optimal} has a simple structure. Uniform information is paired with \(p^*>s/\xi\). At such prices, safe demand cannot be created, and the affine segment of uniform information serves the hedging role. Full and mixture information are paired with \(p^{**}=s/\xi\). At this price, only posterior \(1\) creates safe demand, and both forms of information place positive mass there. Both therefore create positive safe demand, and I refer to them as \textbf{deterrence policies}.

A shared feature of the robustly optimal information provision policies in \autoref{theorem:optimal} is that they all reach the highest posterior, \(1\). 
The intuition is that stopping below \(1\) leaves unused room at the top: moving the upper end of the distribution over posteriors upward raises Buyer's willingness to pay after favorable posterior realizations, which can make purchase without search more likely or support a higher price.
Since the expected posterior must remain equal to the prior \(\mu\), this upward movement requires the remaining posteriors to be lower on average. 
The cost of this lowering is limited, however, because it can come from already unfavorable posteriors: Buyer may not have purchased after observing them anyway.

\begin{figure}[t]
\centering
\begin{tikzpicture}
\begin{axis}[
    width=0.98\linewidth,
    height=0.28\linewidth,
    xmin=0, xmax=1.05,
    ymin=0, ymax=0.25,
    axis lines=left,
    axis line style={black, line width=0.75pt, -{Stealth[length=5.8pt,width=6.2pt]}},
    xlabel={},
    ylabel={},
    xtick={0.1,0.2,0.3,0.4,0.5,0.6,0.7,0.8,0.9,1},
    ytick={0.1,0.2},
    tick style={black, line width=0.75pt},
    tick label style={font=\footnotesize},
    enlargelimits=false,
    clip=false,
    domain=0:1,
    samples=350,
]

\pgfmathsetmacro{\xstar}{sqrt(2)-1}

\addplot[name path=axiszero, draw=none] {0};
\addplot[name path=topcap, draw=none, domain=0:1] {min(x,0.25)};
\addplot[name path=Bone, draw=none, domain=0:1] {x*(1-x)^2/(1+x^2)};
\addplot[name path=Btwo, draw=none, domain=0:1] {x*(1-x)^2/((1+x)^2)};
\addplot[name path=Bthreeleft, draw=none, domain=0:\xstar] {x-2*x^2};

\addplot[draw=none, fill=myblue] fill between[of=axiszero and Btwo, soft clip={domain=0:1}];
\addplot[draw=none, fill=mymaroon] fill between[of=Btwo and Bthreeleft, soft clip={domain=0:\xstar}];
\addplot[draw=none, fill=mygreen] fill between[of=Bthreeleft and Bone, soft clip={domain=0:\xstar}];
\addplot[draw=none, fill=mygreen] fill between[of=Btwo and Bone, soft clip={domain=\xstar:1}];
\addplot[draw=none, fill=mypurple] fill between[of=Bone and topcap, soft clip={domain=0:1}];

\addplot[curvegray, line width=0.45pt, domain=0:1] {x*(1-x)^2/(1+x^2)};
\addplot[curvegray, line width=0.45pt, domain=0:1] {x*(1-x)^2/((1+x)^2)};
\addplot[curvegray, line width=0.30pt, domain=0:\xstar] {x-2*x^2};
\addplot[black, line width=0.75pt, domain=0:0.25] {x};

\node at (axis cs:0.12,0.178) {\(s=\xi\)};
\node[font=\scriptsize, fill=white, fill opacity=0.80, text opacity=1, inner sep=1.2pt, rounded corners=0.8pt]
    at (axis cs:0.56,0.083) {\(B_1(\xi)\)};
\node[font=\scriptsize, fill=white, fill opacity=0.80, text opacity=1, inner sep=1.2pt, rounded corners=0.8pt]
    at (axis cs:0.63,0.033) {\(B_2(\xi)\)};
\node[font=\scriptsize, fill=white, fill opacity=0.80, text opacity=1, inner sep=1.2pt, rounded corners=0.8pt]
    at (axis cs:0.305,0.114) {\(B_3(\xi)\)};
\node[font=\footnotesize, anchor=north east] at (axis cs:0,0) {\(0\)};
\node[anchor=west] at (axis description cs:1.018,-0.11) {\(\xi\)};
\node[anchor=east] at (axis description cs:-0.035,1.02) {\(s\)};

\end{axis}
\end{tikzpicture}
\caption{Regions induced by the cutoffs \(B_1(\xi)\), \(B_2(\xi)\), and \(B_3(\xi)\) in the \((\xi,s)\)-plane.}
\floatfoot{Blue: uniform information. Purple: full information. Maroon: uniform information for low \(\mu\), mixture information for high \(\mu\). Green: uniform information for low \(\mu\), full information for high \(\mu\).}
\label{fig:regions}
\end{figure}

\autoref{fig:regions} illustrates how the cutoffs \(B_1,B_2,B_3\) partition the \((\xi,s)\)-plane. The horizontal axis is the mean outside option, \(\xi\), and the vertical axis is the search cost, \(s\); the 45-degree line reflects the maintained assumption \(s<\xi\). The lower cutoff \(B_2\) and upper cutoff \(B_1\) identify the regions in which the same form of information is optimal for every prior mean match value \(\mu\): below \(B_2\), uniform information is optimal, while above \(B_1\), full information is optimal. Between \(B_2\) and \(B_1\), the optimal form of information provision also depends on \(\mu\). The cutoff \(B_3\) separates the two intermediate regions: in the green region, high \(\mu\) leads to full information; in the maroon region, high \(\mu\) leads to mixture information. In both intermediate regions, low \(\mu\) leads to uniform information.

When search costs are low, corresponding to the blue region in \autoref{fig:regions}, the prices at which safe demand can be generated are too low. Seller instead charges the higher price \(p^*>s/\xi\). At such a price, information cannot generate safe demand, so Seller uses uniform information for hedging. Here, information is noisy not because noise is intrinsically valuable, but because the affine segment of \(H\) makes demand insensitive to the unknown shape of Buyer's outside-option distribution.

When search costs are high, corresponding to the purple region, the restriction \(p\le s/\xi\) required to create safe demand still permits a relatively high price. Full information is then optimal because it maximizes safe demand: when the match value is high, Buyer receives posterior \(1\) and purchases without search for every possible outside-option distribution.\footnote{Subject to the constraint that the expected posterior equals the prior, no information provision policy can assign more than \(\mu\) probability to posterior \(1\). Full information attains this bound by assigning posterior \(1\) when \(x=1\) and posterior \(0\) when \(x=0\).} 
Because these protected purchases occur at the relatively high price \(p^{**}=s/\xi\), maximizing their probability through full information is more profitable than giving up some of them in exchange for return-after-search demand from intermediate posteriors.

For intermediate search costs, corresponding to the green and maroon regions, the optimal form of information provision depends on the prior mean match value \(\mu\). Uniform information allows Seller to charge the higher price \(p^*>s/\xi\), while full and mixture information create safe demand. When \(\mu\) is low, the product is unlikely to be a good match, so very favorable posteriors can arise only with limited probability. The amount of safe demand that can be created is therefore small, making the higher price under uniform information more valuable. When \(\mu\) is high, very favorable posteriors can occur often enough that the resulting safe demand is large enough to justify the lower price.

The maroon and green regions differ in what Seller does once safe demand is worth pursuing.
In the maroon region, \(\xi\) is relatively low, meaning that if Buyer searches, she is relatively likely to find the outside option unattractive. Return-after-search demand is therefore valuable. Mixture information combines the two roles: the mass point at posterior \(1\) generates safe demand, while the uniform part retains the hedging role over intermediate posteriors.
In the green region, \(\xi\) is relatively high, so return-after-search demand is less reliable. At high \(\mu\), Seller therefore gives up the return-after-search channel and uses full information to maximize safe demand.

The search-cost cutoffs are hump-shaped in \(\xi\) because a higher outside-option mean has two opposing effects. It lowers the highest price at which Seller can create safe demand, \(p^{**} = s/\xi\), making this move more costly. At the same time, it makes return-after-search demand less reliable, strengthening the value of safe demand. When \(\xi\) is relatively small, the first effect dominates; when \(\xi\) becomes large, the second effect takes over.

Finally, the ability to choose price has a sharp implication for information provision: unlike in many settings without prices or with exogenous prices, no information is \emph{always suboptimal}. Although no information may induce Buyer to purchase without search with probability one, doing so requires too large a price concession. In particular, full information paired with \(p^{**}=s/\xi\) strictly dominates no information.

\subsubsection{Two Benchmarks} \label{sss:2b}
The two key features of the main model are search frictions and Seller's robustness concerns. The following two benchmarks separate their roles, each retaining only one of these features.

\paragraph{The Role of Search Frictions: Zero Search Cost.} When search is costless, Buyer can inspect the outside option at no cost before deciding whether to buy. The search-deterrence role of information therefore disappears: a mass point ``at the top'' cannot protect demand from search. The hedging role remains, however. As a result, the robustly optimal selling strategy takes the same form as in the uniform-information region of \autoref{theorem:optimal}: Seller uses uniform information and charges the price \(p^*_0 \coloneqq \lim_{s \to 0} p^*\), where \(p^*\) is defined in \eqref{eq:no-det-price}.
This is formalized by \autoref{zsb}; see \Cref{zsc}.

\paragraph{The Role of Robustness Concerns: Known Outside-Option Distribution.}
Now suppose that the outside-option distribution \(G\) is known to Seller, eliminating the robustness concerns.\footnote{For simplicity, assume that \(G\) has full support and admits a strictly log-concave density \(g\).} Search frictions remain, but the hedging motive disappears because there is no adversarial choice of \(G\) to hedge against. In this benchmark, full information is always optimal. This is formalized by \cref{proposition:bayesian}; see \Cref{boo}.

Intuitively, inducing an intermediate posterior \(w\in(0,1)\) amounts to pooling the high and low match-value realizations. 
Relative to full information, this can make sales possible even when the match value is low, but only by giving Buyer a lower expected value than she would have after learning that the match value is high.
Seller must therefore either charge less to induce purchase without search or, at a given price, accept a lower probability that Buyer returns after search. When \(G\) is known to Seller, these losses dominate the benefit from such pooling. Consequently, Seller cannot improve on full information by inducing intermediate posteriors. Thus, with known \(G\), search frictions affect the source of revenue --- purchases without search or purchases after search --- but not the optimality of full information.

Robustness concerns change this conclusion. In the known-\(G\) benchmark, Seller's strategy determines return-after-search demand; in the main model, that demand is determined adversarially through \(G\). When the price is high enough that Buyer may still search even after learning that the match value is high, i.e., when \(p>s/\xi\), uniform information can outperform full information because its affine segment pins down demand rather than letting it vary with the adversarially chosen \(G\). 
When \(p\le s/\xi\), learning that the match value is high creates safe demand, but this price bound limits the revenue from such purchases.
It may therefore be profitable to give up some safe demand in exchange for return-after-search demand from the affine segment. Mixture information can then outperform full information.

\subsection{Comparative Statics} \label{ss:cs}

I conclude this section by discussing comparative statics of the robust price, the robust information provision policy, and the revenue guarantee, together with an implication for Buyer welfare.

\begin{theorem}[Comparative statics] \label{theorem:cs} ~
	\begin{itemize}
		\item [(i)] The robust price \(p_r\) is non-monotone in the search cost \(s\): holding \(\mu\) and \(\xi\) fixed, there exists a cutoff \(\hat{s}\in(0,\xi)\) such that \(p_r\) is increasing on \((0,\hat{s})\) and \((\hat{s},\xi)\), but \( p_r(\hat{s}-) > p_r(\hat{s}+)\). 
        \item[(ii)] For any \(s_1 < s_2\), the robust information provision policy corresponding to \(s_2\) is more informative (in the Blackwell sense) than that corresponding to \(s_1\), unless \(s_1, s_2 \in (B_2(\xi), B_3(\xi))\) and \(\mu\) is sufficiently large.
		\item [(iii)] Seller's revenue guarantee is strictly increasing in \(s\), strictly decreasing in \(\xi\), and increasing in \(\mu\).
	\end{itemize}
\end{theorem}

Part (i) of \autoref{theorem:cs} states that the robust price is not always increasing in the search cost; this is illustrated in \autoref{fig:price-cs}. A higher search cost weakens Buyer's search option, so one might expect Seller to raise the price. This force is present: both \(p^*\), the robust price under uniform information, and \(p^{**}=s/\xi\), the robust price under a deterrence policy, increase with \(s\). The non-monotonicity comes from the switch between uniform information and a deterrence policy. When \(s\) is low, Seller uses uniform information and charges \(p^*>s/\xi\). As \(s\) rises, \(p^{**}\) grows faster than \(p^*\), so the price concession required to generate safe demand becomes smaller, eventually making it more attractive to give up the higher price \(p^*\) in exchange for the safe demand generated by a deterrence policy. Once a deterrence policy becomes optimal, the price drops from \(p^*\) to \(p^{**}\), generating the downward jump in the robust price.\footnote{\cite{wang2017advertising} obtains a related downward price jump in a model in which Buyer searches for more precise information about Seller's product. There, information provision is restricted to a one-parameter truth-or-noise family, as in \cite{lewis1994supplying}. Seller switches from high-price partial information, under which Buyer buys immediately only after favorable posteriors, to no information and a lower price that deters search altogether. \cite{koessler2024sources} show that the jump disappears once arbitrary information provision policies are allowed, because Seller can deter search on path while changing the price continuously. My model differs because Buyer searches over an outside option. Information provision affects Buyer's belief about the match value, but cannot reveal the outside option; thus, even a favorable posterior need not prevent search unless the price is low enough. The relevant switch is therefore from a high price at which favorable posteriors do not create safe demand to a lower price at which they do.}

\begin{figure}
    \centering
    \begin{tikzpicture}[xscale=14.7,yscale=9,>={Stealth[length=2mm,width=1.5mm]}]
        \draw [<->, thick] (0,0.55) node[left]{\(p_r\)} -- (0,0.1) node[below left]{\(O\)} -- (0.3,0.1) node[below]{\(s\)};
        \draw[orange, ultra thick, domain=0:0.1] plot (\x, {(1-sqrt(1-2*\x-(0.5-\x)*(0.5-\x)))/(0.5+\x)});
        \draw[orange, ultra thick, domain=0.1:0.25] plot (\x, {\x*2});
        \draw[dashed, thick] (0.1, 0.1) node[below]{\(\hat{s}\)} -- (0.1, 0.55);
    \end{tikzpicture}
    \caption{The robust price is non-monotone in the search cost.} 
    \floatfoot{When \(0 < s < \hat{s}\), uniform information is optimal; when \(\hat{s} \le s < \xi\), Seller uses a deterrence policy.}
    \label{fig:price-cs}
\end{figure}

Part (ii) asserts that the robust information provision policy generally becomes more informative as the search cost increases.
The high-level intuition is that a higher search cost makes favorable posterior realizations more useful to Seller.
As long as uniform information remains optimal, a higher search cost raises the robust price \(p^*\); the associated uniform policy then becomes more informative because demand must be supported by more favorable posterior realizations. 
If a higher search cost leads Seller to switch to a deterrence policy, informativeness also rises because the policy reveals that the match value is high with strictly positive probability.\footnote{The qualifier ``generally'' refers to the exceptional case stated in \autoref{theorem:cs}, where the robust information provision policies corresponding to both \(s_1\) and \(s_2\) are mixture-information policies. In this case, a higher search cost raises the robust price \(p^{**}=s/\xi\), which is also the lower endpoint of the affine part of mixture information. Holding the mean fixed, this reduces the mass point at posterior \(1\) and makes the distribution over posteriors less spread out, so the robust information provision policy becomes less informative in the Blackwell sense.}

Part (iii) has a simple economic intuition. A higher search cost weakens Buyer's search option and therefore raises the revenue guarantee. A higher mean \(\xi\) of the outside option makes Buyer's alternative more attractive on average, lowering the revenue guarantee. Finally, a higher prior mean match value \(\mu\) makes favorable posterior realizations easier to generate and raises the revenue Seller can guarantee.

Together, Parts (i) and (ii) of \autoref{theorem:cs} imply that easier search need not benefit Buyer. To make this statement precise, fix \(\mu\), \(\xi\), and an outside-option distribution \(G\in\mathcal M(\xi)\), and vary the search cost. Let \(W_G(s)\) denote Buyer's ex ante expected payoff under the robustly optimal selling strategy at search cost \(s\), and let \(W_G(\hat{s}-)\) and \(W_G(\hat{s}+)\) denote its one-sided limits at the cutoff in Part (i).

\begin{corollary}[Buyer welfare]\label{cor:buyer-welfare}
Fix \(\mu\), \(\xi\), and \(G\in\mathcal{M}(\xi)\). At the cutoff \(\hat{s}\) in Part (i) of \autoref{theorem:cs},
\[
W_G(\hat{s}+)>W_G(\hat{s}-).
\]
Consequently, if \(s_1<\hat{s}<s_2\) are sufficiently close to \(\hat{s}\), then \(W_G(s_2)>W_G(s_1)\).
\end{corollary}

The result is driven by the discontinuous change in Seller's strategy at \(\hat{s}\), the cutoff at which Seller switches to a deterrence policy. Evaluated at the common search cost \(\hat{s}\), the strategy just above the cutoff gives Buyer a higher payoff because it charges a strictly lower price and provides more informative product-fit information than the strategy just below it. By continuity, this gain persists for search costs sufficiently close on opposite sides of the cutoff, despite the direct cost of the higher search cost.

\section{Conclusion} \label{s:discussion}
In this paper, I studied a seller launching a novel product who sets a price and provides product-fit information before the buyer decides whether to search for an outside option. The seller is uncertain about the outside-option distribution and seeks strategies robust to this uncertainty. In this setting, price and information cannot be understood separately: the price determines both the margin from a sale and how demanding it is to make buying without search attractive, while the information determines, at that price, whether and to what extent the buyer is induced to buy without search, and how effectively the seller hedges against uncertainty about the outside-option distribution. As search costs vary, optimal strategies combine these forces in different ways, generating distinct price--information patterns across new products.
The comparative statics also caution against a simple pro-consumer view of easier search: lower search costs can raise prices, make information noisier, and leave consumers worse off.

\begin{singlespace}
\bibliographystyle{ecta}
\bibliography{robust.bib}
\end{singlespace}

\begin{appendices}
\appendix
\crefalias{section}{appendix}
\crefalias{subsection}{appendix}
\crefalias{subsubsection}{appendix}
\renewcommand\thefigure{\thesection.\arabic{figure}}
\setcounter{figure}{0}    
\section{Proofs and Omitted Details} \label{prf}
\begin{figure}
	\centering
    \begin{minipage}{0.32\textwidth}
        \centering
        \begin{tikzpicture}[scale=0.62,>={Stealth[length=2mm,width=1.5mm]}]
            \draw[<->, thick] (0,5) node[left]{\small \(H\)} -- (0,0) -- (6,0) node[below]{\small \(w\)};
            \node[left] at (0,4.5) {\small 1};
            \draw[thick] (1.25,0) node [below]{\small \(p\)} -- (1.25,0.2);
            \draw[ultra thick, Blue] (0,0) -- (2,0) -- (4.5, 4.5) -- (5.5, 4.5);
            \draw[thick, dashed] (4.5, 4.5) -- (4.5, 0) node[below right]{\small 1};
            \node[above left, Blue] at (3.7,2.8) {\small \(U_{[2\mu-1,1]}\)};
        \end{tikzpicture}
    \end{minipage}\hfill
    \begin{minipage}{0.32\textwidth}
        \centering
        \begin{tikzpicture}[scale=0.62,>={Stealth[length=2mm,width=1.5mm]}]
            \draw[<->, thick] (0,5) node[left]{\small \(H\)} -- (0,0) -- (6,0) node[below]{\small \(w\)};
            \node[left] at (0,4.5) {\small 1};
            \draw[ultra thick, BurntOrange] (0,0) -- (1.25,0) node[below, black]{\small \(p\)} -- (3.5, 4.5) -- (5.5, 4.5);
            \draw[ultra thick, WildStrawberry] (0,0.5) -- (1.25,0.5) -- (4, 4.5) -- (5.5, 4.5);
            \draw[ultra thick, Fuchsia] (0,1) -- (1.25,1) -- (4.5, 4.5) -- (5.5, 4.5);
            \draw[thick, dashed] (4.5, 4.5) -- (4.5, 0) node[below]{\small 1};
        \end{tikzpicture}
    \end{minipage}\hfill
    \begin{minipage}{0.32\textwidth}
        \centering
        \begin{tikzpicture}[scale=0.62,>={Stealth[length=2mm,width=1.5mm]}]
            \draw[<->, thick] (0,5) node[left]{\small \(H\)} -- (0,0) -- (6,0) node[below]{\small \(w\)};
            \node[left] at (0,4.5) {\small 1};
            \draw[ultra thick, ForestGreen] (0,0) -- (1.15,0) node[below, black]{\small \(p\)} -- (2.8,2.5) -- (3.5,2.5) -- (4.5,4.5) -- (5.5,4.5);
            \draw[thick, dashed] (4.5,4.5) -- (4.5,0) node[below]{\small 1};
        \end{tikzpicture}
    \end{minipage}
    \caption{Distributions that can be optimal when \(p > s/\xi\).} 
    \floatfoot{The left panel corresponds to \(\mu > (1+p)/2\). The middle panel corresponds to \(\mu \le (1+p)/2\) and \(p \ge 2s\): the orange curve depicts an optimal distribution with \(\bar{w} < 1\) and no mass point at \(w = 0\), the pink curve one with \(\bar{w} < 1\) and a mass point at \(w = 0\), and the violet curve one with \(\bar{w} = 1\) and a mass point at 0. The right panel gives an indicative distribution for \(\mu \le (1+p)/2\) and \(s/\xi < p < 2s\): its CDF may have two increasing affine segments, one beginning at \(p\) and the other ending at \(1\), with a constant segment between them.}
    \label{hpm}
\end{figure}

\subsection{Omitted Details} \label{ss:omitted-details}
In the definition of \(p^*\) (see \eqref{eq:no-det-price}), \(\ubar{\mu}\) and \(\bar{\mu}\) are given by 
\[
\ubar{\mu} \coloneqq \frac{2-\xi+s-\sqrt{2(\xi-s)-(\xi-s)^{2}}}{2\left(1-\xi+s\right)}, \text{ and } \bar{\mu} \coloneqq 1-\frac{\sqrt{\xi-s}}{2},
\]
respectively. 

Formally, Seller uses \textbf{uniform information} if the distribution over posteriors is \(U_{[2\mu-1,1]}\) when \(\mu > 1-\left(\sqrt{\xi-s}/2\right)\), and it is
    \[
    H^*_h(w)=\left\{\begin{array}{cc}1-\frac{2 \mu}{1+p^*} & \text{if } w \in[0, p^*) \\ 1-\frac{2 \mu}{1-{p^*}^{2}}(1-w) & \text{if } w \in[p^*, 1]\end{array}\right.
    \]
when \(\mu \le 1-\left(\sqrt{\xi-s}/2\right)\). Seller uses \textbf{mixture information} if the distribution over posteriors is
\begin{equation} \label{huh}
    H^*_u(w) = \left\{\begin{array}{ll} 0 & \text{if } w \in[0, p^{**}), \\ \frac{2\xi^2\left(1-\mu\right)}{\left(\xi-s\right)^2}(w-p^{**}) & \text{if } w \in[p^{**}, 1), \\ 1 & \text{if } w = 1. \end{array}\right.
\end{equation}
Finally, Seller uses \textbf{full information} if the distribution over posteriors is 
\begin{equation} \label{eq:binary-pd}
    H_b^*(w) = \begin{cases} 1-\mu & w \in[0, 1), \\ 1 & w =1. \end{cases}
\end{equation}

\subsection{Proof of \autoref{theorem:optimal}}

\subsubsection{Optimal Information Provision Policy for a Fixed Price} \label{oip}
For any fixed price \(p \in [0,1]\), Seller's optimal information provision policy is summarized in \autoref{ppl} and illustrated in \autoref{hpm} and \autoref{lpm}.

\begin{proposition} \label{ppl} 
    First, suppose \(p\geq1-(\xi-s)\). Then every distribution over posteriors is robustly optimal, and \(\Phi(p)=0\). Next, for every \(p\in(s/\xi,2s]\), \(\Phi(p)<\mu s/\xi\); hence, every selling strategy with such a price is strictly dominated by full information at price \(p^{**}=s/\xi\).\footnote{The exact robustly optimal distributions for \(s/\xi<p\leq2s\) are characterized in \href{https://kunzhang.org/file/roboa.pdf}{Supplementary Appendix} D.} It remains to characterize the cases \(p\leq s/\xi\) and \(\max\{s/\xi,2s\}<p<1-(\xi-s)\).

    Consider first \(\max\{s/\xi,2s\}<p<1-(\xi-s)\). If \(\mu>(1+p)/2\), \(U_{[2\mu-1,1]}\) is robustly optimal.\footnote{\(U_{[a,b]}\) denotes the uniform distribution on \([a,b]\).} Suppose instead that \(\mu\leq(1+p)/2\), and define
\[
\bar w\coloneqq
\begin{cases}
1,&\text{if } \xi-s>\frac{(1-p)^2}{2},\\[6pt] \max\!\left\{2\mu-p,\,p+\xi-s+\sqrt{(\xi-s)(\xi-s+2p)}\right\},&\text{if } \xi-s\leq\frac{(1-p)^2}{2}.
\end{cases}
\]
Then
\begin{equation}\label{dhm}
H_{\bar w}(\cdot\mid p)=\left(1-\frac{2\mu}{p+\bar w}\right)\delta_0+\frac{2\mu}{p+\bar w}U_{[p,\bar w]}
\end{equation}
is robustly optimal.

Now suppose that \(p\leq s/\xi\). If \(\mu\geq p+1-s/\xi\), the degenerate distribution \(\delta_\mu\) is robustly optimal. Suppose instead that \(\mu<p+1-s/\xi\). If \(p\geq(1-2\xi)(\xi-s)/(2\xi^2)\), the binary distribution supported on \(\{0,p+1-s/\xi\}\) is robustly optimal. If \(p<(1-2\xi)(\xi-s)/(2\xi^2)\), there are two cases:
\begin{itemize}[noitemsep,topsep=0pt]
\item[(i)] If \(p+(1-s/\xi)/2\leq\mu<p+1-s/\xi\), the distribution
\begin{equation}\label{hhu}
H^h_u(w\mid p)=
\begin{cases}
0,&w\in[0,p),\\
\frac{2\left[\xi^2(p+1-\mu)-s\xi\right]}{(\xi-s)^2}(w-p),&w\in[p,p+1-s/\xi),\\
1,&w\in[p+1-s/\xi,1]
\end{cases}
\end{equation}
is robustly optimal.
\item[(ii)] If \(\mu<p+(1-s/\xi)/2\), there exists \(\bar w\in[2\mu-p,p+1-s/\xi)\) such that \(H_{\bar w}(\cdot\mid p)\), defined in \eqref{dhm}, is robustly optimal.
\end{itemize}
\end{proposition}

\begin{figure}
    \centering
    \begin{minipage}{0.33\textwidth}
        \centering
        \begin{tikzpicture}[scale=0.62,>={Stealth[length=2mm,width=1.5mm]}]
            \draw[<->, thick] (0,5) node[left]{\small \(H\)} -- (0,0) -- (6,0) node[below]{\small \(w\)};
            \draw[thick, dashed] (4, 4.5) -- (0, 4.5) node[left]{\small 1};
            \draw[ultra thick, Blue] (0,0) -- (1,0) node[below, black]{\small \(p\)} -- (4, 3.25);
            \draw[ultra thick, Blue] (4, 4.5) -- (5.5, 4.5);
            \draw[ultra thick, Cerulean] (0, 1.3) -- (4, 1.3);
            \draw[thick, dashed] (4, 4.5) -- (4, 0);
            \node [below, black] at (3.6, 0) {\small \(p+1-s/\xi\)};
        \end{tikzpicture}
    \end{minipage}\hfill
    \begin{minipage}{0.33\textwidth}
        \centering
        \begin{tikzpicture}[scale=0.62,>={Stealth[length=2mm,width=1.5mm]}]
            \draw[<->, thick] (0,5) node[left]{\small \(H\)} -- (0,0) -- (6,0) node[below]{\small \(w\)};
            \node[left] at (0,4.5) {\small 1};
            \draw[ultra thick, BurntOrange] (0,0) -- (1,0) node[below, black]{\small \(p\)} -- (3.5, 4.5) -- (5.5, 4.5);
            \draw[ultra thick, WildStrawberry] (0,0.6) -- (1,0.6) -- (4, 4.5) -- (5.5, 4.5);
            \draw[thick, dashed] (4.2, 4.5) -- (4.2, 0);
            \draw[ultra thick, Cerulean] (0, 2.2) -- (4.2, 2.2);
            \node [below, black] at (3.8, 0) {\small \(p+1-s/\xi\)};
        \end{tikzpicture}
    \end{minipage}\hfill
    \begin{minipage}{0.33\textwidth}
        \centering
        \begin{tikzpicture}[scale=0.62,>={Stealth[length=2mm,width=1.5mm]}]
            \draw[<->, thick] (0,5) node[left]{\small \(H\)} -- (0,0) -- (6,0) node[below]{\small \(w\)};
            \draw[ultra thick, ForestGreen] (0,0) -- (4,0) node[above right, black]{\small \(\mu\)};
            \draw[ultra thick, ForestGreen] (4,4.5) -- (5.5, 4.5);
            \draw[thick, dashed] (4, 4.5) -- (4, 0);
            \draw[thick, dashed] (4, 4.5) -- (0, 4.5) node[left]{\small 1};
            \draw[thick] (0.75, 0.15) -- (0.75,0) node[below]{\small \(p\)};
            \draw[thick] (3, 0.15) -- (3,0);
            \node [below, black] at (2.6, 0) {\small \(p+1-s/\xi\)};
        \end{tikzpicture}
    \end{minipage}

    \caption{Distributions that can be optimal when \(p \le s/\xi\).}
    \floatfoot{The left panel corresponds to the case \(p+1-s/\xi > \mu \ge p+(1-s/\xi)/2\); the blue curve is \(H^h_u\), and the light blue curve is the binary distribution. The middle panel corresponds to \(\mu < p+(1-s/\xi)/2\); again the light blue curve is the binary distribution, and the orange and pink curves correspond to \(H_{\bar w}\) without and with a mass point at \(0\), respectively. The right panel depicts the case \(\mu \ge p+1-s/\xi\).}
    \label{lpm}
\end{figure}

In what follows, I outline the idea behind the proof of \autoref{ppl}; a formal proof can be found in \href{https://kunzhang.org/file/roboa.pdf}{Supplementary Appendix} D. 

Define Buyer's \textbf{effective outside option} as \(z \coloneqq \min\{v,a\}\), and let \(\hat{G}\) denote its CDF.\footnote{That is, \[\hat{G}(z) \coloneqq \left\{\begin{array}{ll} G(z) & \text{if } z < a, \\ 1 & \text{if } z \ge a. \end{array}\right. \vspace{-5pt}\]} It can be shown that \(z \in [0, 1-s/\xi]\), and \(\mathbb{E}_{\hat{G}}[z] = \xi - s\). 
Nature's choice of outside-option distribution affects Seller's payoff only through the induced effective outside-option distribution.\footnote{This is first observed by \cite{a17} and \cite{cdk}.} Indeed, using \eqref{eq:exp-revenue}, Seller's demand can be written as \(\Psi\left(p,H~\big|~\hat G\right)\coloneqq\mathbb E_{\hat G}[1-H(p+z)]\).
Let \(\widehat{\mathcal{M}}(\xi,s)\) denote the set of distributions on \([0,1-s/\xi]\) with mean \(\xi-s\). 
Every effective outside-option distribution generated by some outside-option distribution \(G \in \mathcal{M}(\xi)\) belongs to \(\widehat{\mathcal{M}}(\xi,s)\), although the converse need not hold. I therefore first consider the auxiliary problem
\begin{equation} \label{emm}
\max_{H \in \mathcal{M}(\mu)} \min_{\hat{G} \in \widehat{\mathcal{M}}(\xi,s)} \Psi\left(p, H ~\big|~ \hat{G}\right).
\end{equation}
Because problem \eqref{emm} enlarges Nature's choice set, its value is a lower bound on the demand Seller can guarantee at price \(p\) in the original problem. I solve it by constructing a saddle point of the zero-sum game in which Seller chooses \(H\) to maximize \(\Psi\) and Nature chooses an effective outside-option distribution \(\hat{G}\) to minimize it.

Given Seller's choice of \((p,H)\), Nature's problem in \eqref{emm} is equivalent to\footnote{The maximum in \eqref{nip} is attained: the set \(\widehat{\mathcal M}(\xi,s)\) is weakly compact, and \(z\mapsto H(p+z)\) is upper semicontinuous.}
\begin{equation} \label{nip}
\max_{\hat{G} \in \widehat{\mathcal{M}}(\xi,s)} \int_{0}^{1-\frac{s}{\xi}} H(p+z) \, \mathrm{d}\hat{G}(z).
\end{equation}
Define 
\[G_p(w)\coloneqq\begin{cases}0,&w\leq p,\\ \hat G((w-p)-),&w>p.\end{cases}\]
By Fubini's theorem, Seller's problem, taking Nature's choice as given, can be written as 
\begin{equation}
\label{sip}
\sup_{H\in\mathcal M(\mu)} \int_0^1G_p(w)\,\mathrm dH(w).
\end{equation}
Consequently, the following best-response test identifies a saddle point.
\begin{lemma} \label{csp}
For a fixed \(p\), \(\left(H^*,\hat{G}^*\right)\) is a saddle point of the auxiliary problem \eqref{emm} if and only if
\[
H^* \in \argmax_{H \in \mathcal{M}(\mu)} \int_{0}^{1} G^*_p(w) \, \mathrm{d} H(w), \quad \text{and} \quad \hat{G}^* \in \argmax_{\hat{G} \in \widehat{\mathcal{M}}(\xi,s)} \int_{0}^{1-\frac{s}{\xi}} H^*(p+z) \, \mathrm{d}\hat{G}(z),
\]
where \(G_p^*\) is constructed from \(\hat G^*\) as above.
\end{lemma}

By Corollary 2 in \cite{kg11}, the value of \eqref{nip} is the concave hull of the function \(z\mapsto H(p+z)\) on \([0,1-s/\xi]\), evaluated at its mean \(\xi-s\).\footnote{For a function \(f\) on a compact interval, its \textbf{concave hull}, denoted by \(\widetilde f\), is the smallest upper-semicontinuous concave function that majorizes \(f\).} Denote its value by \(\widetilde{H}(p+\cdot)(\xi-s)\).
Similarly, the supremum in \eqref{sip} is \(\widetilde{G}_p(\mu)\). A maximizing distribution, when one exists, can be supported on points at which the relevant function equals its concave hull.

To solve the auxiliary problem, I first guess a candidate distribution over posteriors \(H^*\) and compute the concave hull of \(z\mapsto H^*(p+z)\). Nature's best responses are precisely the distributions \(\hat G \in \widehat{\mathcal M}(\xi,s)\) satisfying
\[
\int_0^{1-s/\xi}H^*(p+z)\,\mathrm d\hat G(z)
=\widetilde{H^*}(p+\cdot)(\xi-s).
\]
I then construct such a distribution \(\hat G^*\) for which
\(\int_0^1G_p^*(w)\,\mathrm dH^*(w) =\widetilde{G_p^*}(\mu)\),
so that \(H^*\) is also a best response to \(\hat G^*\). By
\autoref{csp}, the resulting pair is a saddle point of the auxiliary problem \eqref{emm}. Constructing these saddle points across all parameter regions is a central difficulty in proving \autoref{ppl}. Both the candidate distribution over posteriors and Nature's corresponding best response vary across regions; in some cases, the proof must compare a family of candidates and their mixtures rather than verify a single candidate. 

Finally, it remains to connect the auxiliary problem back to Seller's original fixed-price problem. \href{https://kunzhang.org/file/roboa.pdf}{Supplementary Appendix} D shows that \(\hat G^*\) can be induced by a lottery over feasible outside-option distributions. A lottery is sufficient because allowing Nature to randomize does not change the value of her problem. As noted in \Cref{fn:tie-breaking}, the infimum defining Seller's revenue guarantee need not be attained; this possibility is handled with a limiting argument. Together, these arguments close the relaxation and imply that \(H^*\) is robustly optimal in the original problem. This verification is involved but mostly technical, and hence I omit the details here. 

\subsubsection{Two Preliminary Results Regarding the Robust Price}

I first establish two preliminary results, \autoref{claim:high-price} and \autoref{claim:low-price}, that concern the cases of \(p > s/\xi\) and \(p \le s/\xi\), respectively.

\begin{claim}\label{claim:high-price}
Let
\begin{equation}\label{php}
\Pi_h\coloneqq
\begin{cases}
\mu\left(1-\sqrt{2(\xi-s)-(\xi-s)^2}\right),&\mu\leq\ubar{\mu},\\
(2\mu-1)\left(1-\dfrac{\xi-s}{2-2\mu}\right),&\ubar{\mu}<\mu\leq\bar{\mu},\\
\left(1-\sqrt{\xi-s}\right)^2,&\mu>\bar{\mu}.
\end{cases}
\end{equation}
For every \(p\) satisfying \(\max\left\{s/\xi,2s\right\}<p<1-(\xi-s)\), Seller's revenue guarantee satisfies \(\Phi(p)\leq\Pi_h\). 
If \(\Pi_h>\mu s/\xi\), then
\[
\max\left\{\frac{s}{\xi},2s\right\}<p^*<1-(\xi-s),
\]
and \(p^*\) is optimal among prices above \(s/\xi\), with \(\Phi(p^*)=\Pi_h\). At price \(p^*\), if \(\mu>1-\sqrt{\xi-s}/2\), the distribution \(U_{[2\mu-1,1]}\) is robustly optimal. If \(\mu\leq1-\sqrt{\xi-s}/2\), the distribution in \eqref{dhm} with \(p=p^*\) and \(\bar w=1\) is robustly optimal.
\end{claim}

\begin{proof}
Consider any price \(p\) satisfying \(\max\{s/\xi,2s\}<p<1-(\xi-s)\). Writing
\[
t\coloneqq\frac{\xi-s+\sqrt{(\xi-s)(\xi-s+8\mu)}}{4},
\]
Seller's revenue guarantee can be calculated using \autoref{ppl}:
	\begin{equation*}
		\Phi(p) = \left\{\begin{array}{ll} p\left(1-\frac{\xi-s}{1-p}\right) & \text{if } p < 2\mu-1, \\ p\left[1-\frac{\xi-s}{2(\mu-p)}\right]
 & \text{if } \mu - t > p \ge 2\mu-1, \\ \frac{p\mu}{\sqrt{(\xi-s)(\xi-s+2p)}+(\xi-s+p)} & \text{if } 1-\sqrt{2\left(\xi-s\right)} > p \ge \max\{\mu-t, 2\mu-1\}, \\ \frac{2\mu p}{1+p}\left(1-\frac{\xi-s}{1-p}\right) & \text{if } p \ge \max\{1-\sqrt{2\left(\xi-s\right)},2\mu-1\}, \end{array}\right.
	\end{equation*}
The second expression is increasing in \(p\) and equals the third expression at \(p=\mu-t\). The third expression is also increasing in \(p\) and equals the fourth expression at \(p=1-\sqrt{2(\xi-s)}\). Therefore, neither the second nor the third expression can yield a value above the maximum of the fourth expression. To obtain an upper bound, I temporarily drop the two restrictions \(p>\max\{s/\xi,2s\}\) and \(p<1-(\xi-s)\) that define the price interval in the claim. I then maximize the first expression subject to \(p\leq2\mu-1\) and the fourth expression subject to \(p\geq\max\{1-\sqrt{2(\xi-s)},2\mu-1\}\). 
The resulting maximizers are
\[
\begin{cases}
1-\sqrt{\xi-s},&\mu\geq\bar{\mu},\\
2\mu-1,&\mu<\bar{\mu},
\end{cases}
\qquad\text{and}\qquad
\begin{cases}
\frac{1-\sqrt{2(\xi-s)-(\xi-s)^2}}{1-\xi+s},&\mu\leq\ubar{\mu},\\
2\mu-1,&\mu>\ubar{\mu}.
\end{cases}
\]
Evaluating the first and fourth expressions at these maximizing prices and comparing the resulting values gives precisely the three branches of \(p^*\) in \eqref{eq:no-det-price}; the larger of the two values is \(\Pi_h\) in \eqref{php}. Thus, the first and fourth expressions cannot exceed \(\Pi_h\), while the preceding argument shows that the second and third expressions cannot exceed the maximum of the fourth expression. Therefore, \(\Phi(p)\leq\Pi_h\) for every \(\max\left\{s/\xi,2s\right\}<p<1-(\xi-s)\).

Now suppose that \(\Pi_h>\mu s/\xi\). Direct algebra indicates that \(p^* > s/\xi\) and \(p^* < 1-(\xi-s)\); it remains to show that \(p^*>2s\). First consider \(\mu\leq\bar{\mu}\), and let \(H^*\) be the distribution in \eqref{dhm} with \(p=p^*\) and \(\bar w=1\). This distribution is feasible on the first two branches of \eqref{eq:no-det-price}. For every \(\hat G\in\widehat{\mathcal M}(\xi,s)\), Jensen's inequality and \(\mathbb E_{\hat G}[z]=\xi-s\) give
\[
p^*\Psi(p^*,H^*\mid\hat G)=\frac{2\mu p^*}{1+p^*}\mathbb E_{\hat G}\left[\max\left\{1-\frac{z}{1-p^*},0\right\}\right]\geq\frac{2\mu p^*}{1+p^*}\left(1-\frac{\xi-s}{1-p^*}\right)=\Pi_h.
\]
If instead \(\mu>\bar{\mu}\), then \(p^*=1-\sqrt{\xi-s}<2\mu-1\), and
\(1-U_{[2\mu-1,1]}(p^*+z)\geq\max\left\{1-z/(1-p^*),0\right\}\).
Hence,
\[
p^*\Psi\!\left(p^*,U_{[2\mu-1,1]}\mid\hat G\right)\geq p^*\mathbb E_{\hat G}\left[\max\left\{1-\frac{z}{1-p^*},0\right\}\right]\geq p^*\left(1-\frac{\xi-s}{1-p^*}\right)=(p^*)^2=\Pi_h.
\]
Because the auxiliary problem gives a lower bound on Seller's original fixed-price problem, these calculations imply \(\Phi(p^*)\geq\Pi_h\). If \(p^*\leq2s\), then \(s/\xi<p^*\leq2s\), and \autoref{ppl} would imply \(\Phi(p^*)<\mu s /\xi\),
contradicting \(\Phi(p^*)\geq\Pi_h>\mu s/\xi\). Therefore, \(p^*>2s\), and hence
\[
\max\left\{\frac{s}{\xi},2s\right\}<p^*<1-(\xi-s).
\]
Thus, \(\Phi(p^*)=\Pi_h\).

Finally, \autoref{ppl} gives \(\Phi(p)<\mu s/\xi<\Pi_h\) whenever \(s/\xi<p\leq2s\), and \(\Phi(p)=0<\Pi_h\) whenever \(p\geq1-(\xi-s)\). Together with the bound on the remaining interval, this proves that \(p^*\) is optimal among prices above \(s/\xi\). The distributions used to establish the lower bound at \(p^*\) attain \(\Pi_h\) and are therefore robustly optimal.
\end{proof}

\begin{claim} \label{claim:low-price}
	Assume \(p \le s/\xi\). The optimal price in this region is always \(p^{**}=s/\xi\). Furthermore, if \(s \ge B_3(\xi) = \xi - 2 \xi^2\), the optimal distribution over posteriors is the binary distribution \(H^*_b\) defined in \eqref{eq:binary-pd}, and the revenue guarantee is \(\Pi = \mu s / \xi\).
	
	Suppose instead \(s < B_3(\xi)\). 
    If \(\mu\geq p^{**}+(1-s/\xi)/2=(1+s/\xi)/2\), the optimal distribution over posteriors is \(H_u^*\) defined in \eqref{huh}, and the revenue guarantee is
	\begin{equation} \label{rtd}
		\Pi = \frac{s}{\xi}-\frac{2 s \xi \left(1-\mu\right)}{\xi-s}.
	\end{equation}
    If \(\mu<(1+s/\xi)/2\), the optimal distribution over posteriors is \(H_{\bar w}(\cdot\mid p^{**})\), which is \eqref{dhm} evaluated at \(p^{**}\), and Seller's revenue guarantee is
     \begin{equation} \label{fke}
     	\Pi = \left\{\begin{array}{ll} p^{**}\left[1-\frac{\xi-s}{2(\mu-p^{**})}\right]
 	& \text{if } \mu > p^{**} \text{ and } \xi - s \le \frac{2(\mu-p^{**})^2}{2 \mu - p^{**}}, \\ \frac{p^{**}\mu}{\sqrt{(\xi-s)(\xi-s+2p^{**})}+(\xi-s+p^{**})} & \text{otherwise}. \end{array}\right.
     \end{equation}
\end{claim}

\begin{proof}[Proof of \autoref{claim:low-price}]
If \(p\geq(1-2\xi)(\xi-s)/(2\xi^2)\), then by \autoref{ppl}, the binary distribution supported on \(\{0,p+1-s/\xi\}\) is robustly optimal whenever \(p>\mu+s/\xi-1\). For such a price, Seller's revenue guarantee is 
\[p\frac{\mu}{p+1-s/\xi}=\frac{p\mu\xi}{\xi(1+p)-s},\] 
which is strictly increasing in \(p\) on this price range. If instead \(0\leq p\leq\mu+s/\xi-1\), \autoref{ppl} implies that the degenerate distribution \(\delta_\mu\) is robustly optimal, and Seller's revenue guarantee is \(p\). At the upper endpoint of this range, it is \(\mu+s/\xi-1\). Full information at price \(p^{**}\) yields \(\mu s/\xi\), and 
\[\left(\mu+\frac{s}{\xi}-1\right)-\frac{\mu s}{\xi}=\left(1-\frac{s}{\xi}\right)(\mu-1)<0.\] 
Therefore, every strategy involving the degenerate distribution is dominated by full information at price \(p^{**}\) and can be ignored henceforth.

When \(p<(1-2\xi)(\xi-s)/(2\xi^2)\) and \(p>\mu-(1-s/\xi)/2\), again by \autoref{ppl}, \(H_{\bar w}\) defined in \eqref{dhm} is robustly optimal.\footnote{If \((1-2\xi)(\xi-s)/(2\xi^2)<\mu-(1-s/\xi)/2\), then \(H_{\bar w}(\cdot\mid p^{**})\) is never optimal.} An argument similar to that in the proof of \autoref{claim:high-price} shows that Seller's revenue guarantee is strictly increasing in \(p\) on this range. For prices satisfying \(\mu+s/\xi-1<p\leq\mu-(1-s/\xi)/2\) and \(p<(1-2\xi)(\xi-s)/(2\xi^2)\), Seller chooses the price to solve
\[
\max_p ~~ p \left[1-2\xi\frac{\xi\left(p+1-\mu\right)-s}{\xi-s}\right].
\]
The objective is strictly concave in \(p\), so it suffices to look at the first-order condition (FOC). The FOC yields
\begin{equation} \label{eq:point}
p^i = \frac{(1-2\xi)(\xi-s)}{4 \xi^{2}}+\frac{\mu}{2};
\end{equation}
I claim, however, that \(p^i \ge \min \{\mu - (1-s/\xi)/2, \left(1-2\xi\right)\left(\xi-s\right)/(2\xi^2)\}\). Using \eqref{eq:point}, simple algebra reveals that the inequality is equivalent to either
	\begin{equation} \label{kfm}
		\mu \ge \frac{\left(1-2\xi\right)\left(\xi-s\right)}{2\xi^{2}}
	\end{equation}
	or
	\begin{equation} \label{lfm}
		\mu \le \frac{\xi-s}{2\xi^2}.
	\end{equation}
    If \eqref{kfm} fails, then its right-hand side must be positive because \(\mu\geq0\), which implies \(\xi<1/2\). Hence \(0<1-2\xi<1\), and 
    \[\mu<(1-2\xi)\frac{\xi-s}{2\xi^2}\leq\frac{\xi-s}{2\xi^2},\] 
    where the second inequality also uses \(s<\xi\). Thus, \eqref{lfm} holds whenever \eqref{kfm} fails. Since \(p^i\) is at least the upper endpoint of the price range in which \(H_u^h\) is robustly optimal and the objective is strictly concave, the corresponding revenue guarantee is strictly increasing in \(p\) throughout that range. The formulas for \(\Phi(p)\) also agree whenever two price ranges meet: at \(p=\mu+s/\xi-1\), both relevant formulas equal \(p\); at \(p=\mu-(1-s/\xi)/2\), both equal \(p(1-\xi)\); and at \(p=(1-2\xi)(\xi-s)/(2\xi^2)\), both equal \(\mu(1-2\xi)\). Together with the strict increase established within each price range, these equalities imply that \(\Phi(p)\) is strictly increasing on \([0,s/\xi]\). Therefore, the optimal price in this region is \(p^{**}=s/\xi\), and the robustly optimal distribution and revenue guarantee are determined by the case containing \(p^{**}\).

Finally, it is straightforward to see that 
\[
p^{**} = \frac{s}{\xi} \ge \frac{\left(1-2\xi\right)\left(\xi-s\right)}{2\xi^{2}}
\]
if and only if \(s \ge B_3(\xi)\). This completes the proof.
\end{proof}

\subsubsection{Final Step}
It now remains to compare the revenue guarantees from the two price regions characterized in \Cref{claim:high-price,claim:low-price}. According to \Cref{ppl,claim:high-price}, whenever \(\Pi_h\leq\mu s/\xi\), any selling strategy with \(p>s/\xi\) is dominated by providing full information and charging \(p^{**}\). Thus, for prices above \(s/\xi\), it suffices to consider the case \(\Pi_h>\mu s/\xi\), in which the highest revenue guarantee is \(\Pi_h\). Moreover, the only candidates from the region \(p\leq s/\xi\) that can be optimal in the unrestricted problem are \((p^{**},H_b^*)\) and \((p^{**},H_u^*)\). Indeed, the only additional possibility identified by \autoref{claim:low-price} is \((p^{**},H_{\bar w}(\cdot\mid p^{**}))\), which arises when \(s<B_3(\xi)\) and \(\mu<(1+s/\xi)/2\). In this case, \(s<B_3(\xi)\) implies \(\xi<1/2\) and \(p^{**}>2s\); the relevant expression for \(\Phi(p)\) in the proof of \autoref{claim:high-price} is therefore strictly increasing immediately to the right of \(p^{**}\). Hence, although \((p^{**},H_{\bar w}(\cdot\mid p^{**}))\) solves the problem restricted to \(p\leq s/\xi\), it cannot be optimal in the unrestricted problem.

For fixed \(\xi\in(0,1)\) and \(s\in(0,\xi)\), define \(D(s)\coloneqq\left\{\mu\in(0,1):\mu s/\xi\ge\Pi_h\right\}\), and
\[
N(s)\coloneqq\left\{\mu\in(0,1):\mu\ge\frac{1+s/\xi}{2}\text{ and }\frac{s}{\xi}-\frac{2s\xi(1-\mu)}{\xi-s}\ge\Pi_h\right\},
\]
where \(\Pi_h\) is defined in \eqref{php}. Thus, \(D(s)\) and \(N(s)\) are the sets of priors for which the revenue guarantees from \((p^{**},H_b^*)\) and \((p^{**},H_u^*)\), respectively, exceed \(\Pi_h\).

Direct algebra shows that each of \(D(s)\) and \(N(s)\) is either empty or consists of all priors above some cutoff. It also shows that both sets are empty if and only if
\begin{equation}\label{nds}
s\le\frac{\xi(\xi-1)^2}{(\xi+1)^2}=B_2(\xi).
\end{equation}
Thus, uniform information is optimal for every \(\mu\in(0,1)\) whenever \eqref{nds} holds. If instead \(s>B_2(\xi)\), both sets are nonempty, and I define
\[
\hat{\mu}(s)\coloneqq\inf D(s),\qquad \check{\mu}(s)\coloneqq\inf N(s).
\]
Because \(\mu s/\xi\) and the first expression in \eqref{php} are linear in \(\mu\), \(\hat{\mu}(s)=0\) if and only if
\begin{equation}\label{pvc}
\frac{s}{\xi}\ge 1-\sqrt{2(\xi-s)-(\xi-s)^2},
\end{equation}
or equivalently,
\begin{equation}\label{mds}
s\ge\frac{\xi(\xi-1)^2}{\xi^2+1}=B_1(\xi).
\end{equation}
If \eqref{mds} holds, full information is optimal for every \(\mu\in(0,1)\). 

If instead \(B_2(\xi)<s<B_1(\xi)\), then \(\hat{\mu}(s),\check{\mu}(s)\in(0,1)\). In this case, the difference between the revenue guarantees from mixture and full information is
\begin{equation}\label{eq:det-comparison}
\left[\frac{s}{\xi}-\frac{2s\xi(1-\mu)}{\xi-s}\right]-\frac{\mu s}{\xi}=\frac{s(1-\mu)\bigl(B_3(\xi)-s\bigr)}{\xi(\xi-s)}.
\end{equation}
For \(\mu>(1+s/\xi)/2\), \eqref{eq:det-comparison} shows that full information yields at least as much revenue as mixture information when \(s\geq B_3(\xi)\), while mixture information yields strictly more when \(s<B_3(\xi)\). Suppose first that \(s\geq B_3(\xi)\). By \autoref{claim:low-price}, full information yields the highest revenue guarantee over prices \(p\leq s/\xi\). If \(\mu<\hat{\mu}(s)\), its revenue guarantee is less than \(\Pi_h\), and uniform information at \(p^*\) is optimal. If \(\mu\geq\hat{\mu}(s)\), full information yields at least \(\Pi_h\) and is optimal. Therefore, if \(s\geq B_3(\xi)\), uniform information is optimal for \(\mu<\hat{\mu}(s)\), while full information is optimal for \(\mu\geq\hat{\mu}(s)\).

Now suppose that \(s<B_3(\xi)\). When \(\mu\leq(1+s/\xi)/2\), the strict increase of \(\Phi(p)\) immediately to the right of \(p^{**}\), together with the fact that full information at \(p^{**}\) guarantees \(\mu s/\xi\), implies that some price above \(s/\xi\) yields a revenue guarantee strictly greater than \(\mu s/\xi\). Hence, \(\Pi_h>\mu s/\xi\), and \autoref{claim:high-price} implies that uniform information at \(p^*\) is optimal. When \(\mu>(1+s/\xi)/2\), \autoref{claim:low-price} shows that \((p^{**},H_u^*)\) yields the highest revenue guarantee over prices \(p\leq s/\xi\). Seller therefore compares the revenue guarantee in \eqref{rtd} with \(\Pi_h\). By the definitions of \(N(s)\) and \(\check{\mu}(s)\), uniform information is optimal for \(\mu<\check{\mu}(s)\), while mixture information is optimal for \(\mu\geq\check{\mu}(s)\).

Whenever the comparisons above select uniform information, \(\mu\notin D(s)\), and hence \(\Pi_h>\mu s/\xi\). Therefore, \autoref{claim:high-price} implies that \(p^*>\max\{s/\xi,2s\}\). This completes the proof.

\subsection{Proof of \autoref{theorem:cs}} \label{icp}
\subsubsection{A Preliminary Result}
The following simple result is useful for the proof. For any \(p\in(0,1)\), define \(H_p=U_{[2\mu-1,1]}\) if \(p\le2\mu-1\), and otherwise define
\[
H_p(w)\coloneqq
\begin{cases}
1-\frac{2\mu}{1+p} & \text{if }w\in[0,p),\\
1-\frac{2\mu}{1-p^2}(1-w) & \text{if }w\in[p,1].
\end{cases}
\]
Whenever uniform information is optimal, the distribution specified in \Cref{ss:omitted-details} is \(H_{p^*}\).

\begin{claim}\label{c:how_p_affects_info}
If \(p_1\le p_2\), then \(H_{p_2}\) is a mean-preserving spread (MPS) of \(H_{p_1}\).\footnote{Let \(F_1\) and \(F_2\) be two distributions defined on \([0,1]\). \(F_1\) is a mean-preserving spread of \(F_2\) if \(\int_0^x F_2(t)\,\mathrm{d}t\le\int_0^x F_1(t)\,\mathrm{d}t\) for all \(x\in[0,1]\), with equality at \(x=1\).}
\end{claim}

\begin{proof}[Proof of \autoref{c:how_p_affects_info}]
If \(p_1=p_2\) or \(p_2\le2\mu-1\), the claim is immediate. Suppose instead that \(p_1<p_2\) and \(p_2>2\mu-1\). Inspection of the definitions shows that \(H_{p_1}\) crosses \(H_{p_2}\) only once, and from below. Since both distributions have mean \(\mu\), Theorem 3.A.44 in \cite{ss07} implies that \(H_{p_2}\) is a MPS of \(H_{p_1}\).
\end{proof}

\subsubsection{Proof of Part (i)}
On each branch of \(\Pi_h\), implicit differentiation of the equality defining the relevant cutoff shows that \(\check{\mu}(s)\) is strictly decreasing on \((B_2(\xi),B_3(\xi))\), whenever this interval is nonempty, and that \(\hat{\mu}(s)\) is strictly decreasing on \((\max\{B_2(\xi),B_3(\xi)\},B_1(\xi))\). The cutoff is continuous when the relevant branch of \(\Pi_h\) changes. The relevant cutoff converges to \(1\) as \(s\searrow B_2(\xi)\), while \eqref{eq:det-comparison} implies that \(\check{\mu}(B_3(\xi))=\hat{\mu}(B_3(\xi))\) whenever \(B_2(\xi)<B_3(\xi)\). Together with the fact that full information is optimal for every \(\mu\in(0,1)\) when \(s\ge B_1(\xi)\), these observations imply that, for fixed \(\mu\) and \(\xi\), there exists a cutoff \(\hat{s}\in(B_2(\xi),B_1(\xi)]\) such that uniform information is used for \(s<\hat{s}\), while a deterrence policy is used for \(s\ge\hat{s}\).

Therefore,
\[
p_r(s)=
\begin{cases}
p^*(s) & \text{if }s\in(0,\hat{s}),\\
p^{**}(s) & \text{if }s\in[\hat{s},\xi).
\end{cases}
\]
It follows directly from \eqref{eq:no-det-price} that \(p^*(s)\) is increasing in \(s\), while \(p^{**}(s)=s/\xi\) is strictly increasing. Thus, \(p_r\) is increasing on \((0,\hat{s})\) and \((\hat{s},\xi)\).

It remains to establish the downward jump at \(\hat{s}\). Since \(\hat{s}\le B_1(\xi)\), the equivalence between \eqref{pvc} and \eqref{mds} implies that
\[
\frac{\hat{s}}{\xi}\le 1-\sqrt{2(\xi-\hat{s})-(\xi-\hat{s})^2}.
\]
Moreover, every branch of \(p^*\) in \eqref{eq:no-det-price} is bounded below by its first branch. Therefore,
\[
p^*(\hat{s})\ge\frac{1-\sqrt{2(\xi-\hat{s})-(\xi-\hat{s})^2}}{1-\xi+\hat{s}}>1-\sqrt{2(\xi-\hat{s})-(\xi-\hat{s})^2}\ge\frac{\hat{s}}{\xi}=p^{**}(\hat{s}),
\]
where the strict inequality follows from \(0<1-\xi+\hat{s}<1\). Since \(p^*\) and \(p^{**}\) are continuous, \(p_r(\hat{s}-)=p^*(\hat{s})>p^{**}(\hat{s})=p_r(\hat{s}+)\).

\subsubsection{Proof of Part (ii)}
Fix \(\mu\) and \(\xi\), and suppose \(s_1<s_2\). Denote the corresponding distributions over posteriors by \(H_{s_1}\) and \(H_{s_2}\), respectively. By Theorem 7 in \cite{b53d}, the information provision corresponding to \(s_2\) is Blackwell more informative than that corresponding to \(s_1\) if and only if \(H_{s_2}\) is a mean-preserving spread of \(H_{s_1}\).

By the cutoff argument in Part (i), if uniform information is optimal at \(s_2\), it is also optimal at \(s_1\). In this case, Part (i) implies that \(p^*(s_1)\le p^*(s_2)\). Since the search cost affects the uniform information specified in \Cref{ss:omitted-details} only through the robust price, \autoref{c:how_p_affects_info} implies that \(H_{s_2}\) is a mean-preserving spread of \(H_{s_1}\).

If \(H_{s_2}\) is full information, it is Blackwell more informative than \(H_{s_1}\). It therefore remains to consider the case in which \(H_{s_2}\) is mixture information. This requires \(\xi<\sqrt{2}-1\), \(B_2(\xi)<s_2<B_3(\xi)\), and \(\mu\ge\check{\mu}(s_2)\).

Suppose first that \(H_{s_1}\) is uniform information. Let \(\mu_0\coloneqq1/(1+\xi)\). Because \(B_2(\xi)<s_2<B_3(\xi)\), direct algebra gives
\[
\frac{1+s_2/\xi}{2}<\mu_0
\qquad\text{and}\qquad
\ubar{\mu}(s_2)<\mu_0<\bar{\mu}(s_2).
\]
Thus, mixture information is feasible at \(\mu_0\), and \(\Pi_h\) is given by the middle branch of \eqref{php}. The difference between the revenue guarantee from mixture information, given by \eqref{rtd}, and \(\Pi_h\) is
\[
\left[\frac{s_2}{\xi}-\frac{2s_2\xi(1-\mu_0)}{\xi-s_2}\right]-(2\mu_0-1)\left(1-\frac{\xi-s_2}{2-2\mu_0}\right)=-\frac{\left[s_2(1+\xi)-\xi(1-\xi)\right]^2}{2\xi(\xi-s_2)(1+\xi)}<0.
\]
Since \(\check{\mu}(s_2)\) is the cutoff above which mixture information becomes optimal, it follows that \(\check{\mu}(s_2)>\mu_0\). Moreover, the expression for \(\ubar{\mu}\) implies that \(s_1<s_2<B_3(\xi)\) yields \(\ubar{\mu}(s_1)<\mu_0\). Hence, \(\mu\ge\check{\mu}(s_2)\) implies that \(\mu>\ubar{\mu}(s_1)\). By \eqref{eq:no-det-price} and the definition in \Cref{ss:omitted-details}, the uniform information corresponding to \(s_1\) is therefore \(U_{[2\mu-1,1]}\).

Because \(s_2\) corresponds to mixture information, \(2\mu-1>s_2/\xi\). It follows directly from \eqref{huh} that \(U_{[2\mu-1,1]}\) crosses \(H_{s_2}\) only once, and from below. Since both distributions have mean \(\mu\), Theorem 3.A.44 in \cite{ss07} implies that \(H_{s_2}\) is a mean-preserving spread of \(U_{[2\mu-1,1]}\). Thus, \(H_{s_2}\) is Blackwell more informative than \(H_{s_1}\).

Finally, suppose that \(H_{s_1}\) is also mixture information. Because \(s_1<s_2\), the lower endpoint \(s_1/\xi\) of its affine segment is below the corresponding endpoint \(s_2/\xi\) under \(H_{s_2}\). It follows directly from \eqref{huh} that \(H_{s_2}\) crosses \(H_{s_1}\) only once, and from below. Since both distributions have mean \(\mu\), Theorem 3.A.44 in \cite{ss07} implies that \(H_{s_1}\) is a mean-preserving spread of \(H_{s_2}\). Hence, the Blackwell ranking is reversed.

Therefore, for any \(s_1<s_2\), the information provision corresponding to \(s_2\) is Blackwell more informative than that corresponding to \(s_1\), except when \(B_2(\xi)<s_1<s_2<B_3(\xi)\) and \(\mu\ge\check{\mu}(s_1)\).
In this exceptional case, both search costs lead to mixture information, and the information provision corresponding to \(s_1\) is Blackwell more informative than that corresponding to \(s_2\).

\subsubsection{Proof of Part (iii)}
\autoref{srg}, which is a corollary of \autoref{theorem:optimal}, \autoref{claim:high-price}, and \autoref{claim:low-price}, summarizes Seller's revenue guarantee for different parameter values.

\begin{claim} \label{srg}
If \(s\ge B_1(\xi)\), Seller's revenue guarantee is \(\Pi=\mu s/\xi\); and if \(s\le B_2(\xi)\), Seller's revenue guarantee is given by \eqref{php}. If \(B_2(\xi)<s<B_1(\xi)\), there are two cases:
\begin{itemize}
    \item[(1)] If either \(\xi\ge\sqrt{2}-1\), or \(\xi<\sqrt{2}-1\) and \(B_3(\xi)\le s<B_1(\xi)\), there exists \(\hat{\mu}\in(0,1)\) such that for \(\mu<\hat{\mu}\), Seller's revenue guarantee is given by \eqref{php}; and for \(\mu\ge\hat{\mu}\), Seller's revenue guarantee is \(\Pi=\mu s/\xi\).
    \item[(2)] If instead \(\xi<\sqrt{2}-1\) and \(B_2(\xi)<s<B_3(\xi)\), there exists \(\check{\mu}\in(0,1)\) such that for \(\mu<\check{\mu}\), Seller's revenue guarantee is given by \eqref{php}; and for \(\mu\ge\check{\mu}\), Seller's revenue guarantee is given by \eqref{rtd}.
\end{itemize}
\end{claim}

Part (iii) follows by differentiating the revenue expressions in \autoref{srg} within each parameter region and noting that adjacent expressions coincide at the cutoffs.

\subsection{Proof of \autoref{cor:buyer-welfare}}\label{proof:buyer-welfare}

Fix \(G\in\mathcal{M}(\xi)\). For any selling strategy \((p,H)\) and search cost \(s\), let
\[
U_G(p,H;s)\coloneqq\int_0^1 \max\left\{0,w-p,\mathbb{E}_G[\max\{w-p,v\}]-s\right\}\,\mathrm{d}H(w)
\]
denote Buyer's ex ante expected payoff. Let \((p^-,H^-)\) and \((p^+,H^+)\) denote the limits of the robustly optimal selling strategies as \(s\nearrow\hat{s}\) and \(s\searrow\hat{s}\), respectively. Parts (i) and (ii) of \autoref{theorem:cs}, together with continuity of the policy formulas on either side of \(\hat{s}\), imply that \(p^->p^+\) and that \(H^+\) is a mean-preserving spread of \(H^-\).

For fixed \(p\), \(s\), and \(G\), Buyer's optimal expected payoff conditional on posterior \(w\),
\[
\max\left\{0,w-p,\mathbb{E}_G[\max\{w-p,v\}]-s\right\},
\]
is convex in \(w\). Hence, \(U_G(p^+,H^+;\hat{s})\ge U_G(p^+,H^-;\hat{s})\).
Moreover, Seller's revenue under \((p^-,H^-)\) is positive for every \(G\in\mathcal{M}(\xi)\), since it is bounded below by her positive revenue guarantee. Buyer therefore purchases Seller's product with positive probability. If the price is reduced from \(p^-\) to \(p^+\) and Buyer follows the same strategy, her payoff rises by \(p^--p^+\) whenever she purchases Seller's product and is unchanged otherwise. Thus,
\(U_G(p^+,H^-;\hat{s})>U_G(p^-,H^-;\hat{s})\).
Combining the two inequalities gives \(U_G(p^+,H^+;\hat{s})>U_G(p^-,H^-;\hat{s})\).
Buyer's expected payoff varies continuously with \(s\) along each of the two strategy branches. Therefore, \(W_G(\hat{s}+)>W_G(\hat{s}-)\), and the same comparison holds for \(s_1<\hat{s}<s_2\) sufficiently close to \(\hat{s}\).

\section{Two Benchmarks: Formal Statements}
\subsection{Zero Search Cost} \label{zsc}
When \(s = 0\), Buyer always inspects the outside option since it is costless to do so. \autoref{zsb} describes the robustly optimal selling strategy in this case.

\begin{proposition} \label{zsb}
    Suppose \(s = 0\). Uniform information is optimal, and the robust price is \(p^*_0\). 
\end{proposition}

\begin{proof}
    When \(s = 0\), the definition of \(a\) indicates that \(a = 1\). Consequently, Seller's payoff when the distribution over outside options is \(G\) can be simplified to \(p \,\mathbb{E}_G[1-H(p+v)]\). Therefore, I can work with the outside-option distribution directly; in particular, there is no need to find an effective outside-option distribution first and then find an outside-option distribution that generates it.

    For each pair consisting of a distribution over posteriors and an effective outside-option distribution \(\left(H, \hat{G}\right)\) that is a saddle point in the proof of Proposition 1, it can be checked that by setting \(s=0\), the resulting pair \((H', G')\) of the distribution over posteriors and outside-option distribution is a saddle point for the zero search cost problem. Using these results, an argument similar to the proof of Claim 1 establishes the proposition.
\end{proof}

\subsection{Known Outside-Option Distribution} \label{boo}
The only difference between this benchmark and the main model is that Seller \emph{knows} Buyer's outside-option distribution \(G\). For simplicity, assume that \(G\) has full support and admits a strictly log-concave density \(g\).

\begin{proposition} \label{proposition:bayesian}
    An optimal selling strategy fully reveals the match value and charges
    \[
    p^o = 
    \begin{cases}
     1- a & \text{ if  } 1-a \ge p_hG\left(1-p_h\right), \\
     p_h & \text{ if  } 1-a < p_hG\left(1-p_h\right),
    \end{cases}
    \]
    where \(a\) is defined in \eqref{eq:res-value}, and \(p_h\) is the unique maximizer of \(pG(1-p)\), characterized by\footnote{Uniqueness follows because \(G\) has full support and \(g\) is strictly log-concave.}
\begin{equation} \label{dfp}
    p = \frac{G(1-p)}{g(1-p)}.
\end{equation}
\end{proposition}

The optimality of full information follows from the observation that, absent full information, Seller can always increase her profits by either raising the price or providing more information, or both. To understand the pricing formula in \autoref{proposition:bayesian}, suppose first that Seller charges \(1-a\) under full information. When the match value is high, Buyer receives posterior \(1\) and purchases without search; when the match value is low, Buyer receives posterior \(0\) and does not buy from Seller. Seller's revenue is therefore \(\mu(1-a)\). Alternatively, Seller can charge a price \(p\) at which Buyer searches after posterior \(1\). In that case, Buyer returns only when the outside option realization is below \(1-p\), so Seller's revenue is \(\mu pG(1-p)\). The best such price is \(p_h\). Hence Seller compares \(\mu(1-a)\) with \(\mu p_hG(1-p_h)\): when the former is larger, Seller charges \(1-a\), so Buyer purchases without search after the high match-value realization; otherwise, Seller charges \(p_h\) and accommodates Buyer's search. See \Cref{proof:bayesian} for the formal proof. 

The search-accommodating price is selected only if \(1-a<p_hG(1-p_h)\), which implies \(p_h>1-a\). Thus, the model features a familiar trade-off: deterring search requires a lower price, while accommodating search permits a higher price but exposes demand to the outside option. For a fixed \(G\), this trade-off is governed entirely by the search cost.

\begin{corollary} \label{scf}
    For every outside-option distribution \(G\), there exists \(\hat{s}_G \in (0, \xi)\) such that \(p^o = p_h\) for every \(s < \hat{s}_G\), and \(p^o = 1-a\) for every \(s \ge \hat{s}_G\). Furthermore, at \(s = \hat{s}_G\), the optimal price drops from \(p_h\) to \(1 - a\left(\hat{s}_G\right)\). 
\end{corollary}

\begin{proof}
    Because the outside-option distribution \(G\) has full support, \(a\) is strictly decreasing in \(s\). Because \(p_h\) does not depend on \(s\), there exists a unique \(a^*\) solving \(1-a = p_h G\left(1-p_h\right)\). Let \(\hat{s}_G\) denote the search cost corresponding to \(a^*\). The statement then follows from \autoref{proposition:bayesian}. 
\end{proof}

\subsubsection{Proof of \autoref{proposition:bayesian}} \label{proof:bayesian}
For a fixed posterior \(w\), the probability that Buyer buys from Seller is given by \(G_p^a(w)\), where, when \(p+a\leq1\),
\[
G_p^a(w)\coloneqq\left\{\begin{array}{ll} 0 & \text{ if }  w < p, \\ G(w-p) & \text{ if }  p \le w < p+a, \\ 1 & \text{ if } w \ge p+a; \end{array}\right.
\]
and
\[
G_p^a(w)\coloneqq\left\{\begin{array}{ll}0 &   \text{ if }  w < p, \\ G(w-p) & \text{ if } w \ge p, \end{array} \right.
\]
when \(p+a > 1\). Therefore, if Seller chooses a distribution over posteriors \(H\), her payoff can be written as \(p\int_0^1G_p^a(w)\,\mathrm dH(w)\). Consequently, Seller's problem is
\[
\max_{p \in [0,1]} \left\{\max_{H \in \mathcal{M}(\mu)} p \int_0^1 G^a_p(w) \, \mathrm{d}H(w)\right\}.
\]

It is useful to solve for the optimal information provision policy for a fixed price first. For a fixed \(p \in [0,1]\), Seller's problem of choosing a distribution over posteriors is
\[
\max_{H \in \mathcal{M}(\mu)} \int_0^1 G^a_p(w) \, \mathrm{d}H(w),
\]
which is identical to the information design problem studied in \cite{kg11},\footnote{See the problem on page 2596 in \cite{kg11}.} where Seller, who plays the role of Sender in their framework, has value function \(G^a_p(w)\). Consequently, Seller's optimal distribution can be identified by finding the concave hull of \(G^a_p(w)\).

To identify the concave hull of \(G_p^a\), for \(p>0\), let \(r(p)\in(p,1]\) solve
\begin{equation} \label{dfr}
g(r(p) - p) r(p) = G(r(p) - p).
\end{equation}
If no such solution exists, set \(r(p)=1\). For \(p=0\), set \(r(0)=0\) if \(G\) is concave; otherwise, let \(r(0)\) be a positive solution to \(g(r(0))r(0)=G(r(0))\), if one exists, and set \(r(0)=1\) otherwise. Note that when \(r(p)<1\), it must be that \(g'(r(p)-p)<0\). If \(p>1-a\), \(G_p^a\) is concave on \([r(p),1]\), and its concave hull is affine on \([0,r(p)]\) and equals \(G_p^a\) on \([r(p),1]\).

Now suppose \(p\leq1-a\). Let \(\ell_p\) denote the line segment joining \((0,0)\) and \((p+a,1)\). If \(g(r(p)-p)\leq1/(p+a)\), the concave hull of \(G_p^a\) equals \(\ell_p\) on \([0,p+a]\) and equals \(1\) on \([p+a,1]\). If instead \(g(r(p)-p)>1/(p+a)\), let \(t(p)\) be the point at which the line through \((p+a,1)\) is tangent to \(G_p^a\). It solves
\begin{equation}\label{dft}
g(t(p)-p)(p+a-t(p))=1-G(t(p)-p).
\end{equation}
The concave hull of \(G_p^a\) is then affine on \([0,r(p)]\) and \([t(p),p+a]\), equals \(G_p^a\) on \([r(p),t(p)]\), and equals \(1\) on \([p+a,1]\).

Once the concave hull of \(G^a_p\) is identified, the optimal information provision for a fixed price follows immediately.

\begin{lemma}\label{ifp}
    Suppose that \(p\leq1-a\). If \(g(r(p)-p)\leq1/(p+a)\), then \(\{0,p+a\}\) is optimal when \(\mu\in(0,p+a)\),\footnote{Since every optimal distribution over posteriors is either degenerate or binary, I identify such a distribution by its support.} while \(\{\mu\}\) is optimal when \(\mu\in[p+a,1]\). If \(g(r(p)-p)>1/(p+a)\), 
    \begin{itemize}
    \item when \(\mu \in(0, r(p))\), \(\{0, r(p)\}\) is optimal;
    \item when \(\mu \in[r(p), t(p)]\), \(\{\mu\}\) is optimal;
    \item when \(\mu \in(t(p), p+a)\), \(\{t(p), p+a\}\) is optimal; and
    \item when \(\mu \in [p+a,1)\), \(\{\mu\}\) is optimal.
    \end{itemize}
    
    Suppose instead that \(p > 1-a\). Then when \(\mu \in(0, r(p))\), \(\{0, r(p)\}\) is optimal;
    and when \(\mu \in[r(p), 1)\), \(\{\mu\}\) is optimal.
\end{lemma}

\begin{proof}[Proof of \autoref{proposition:bayesian}]
    Suppose first that \(p > 1-a\). Let
\[
\hat p\coloneqq\sup\left\{p\in[0,1]:\eqref{dfr}\text{ has a solution in }(p,1]\right\}.
\]
By the implicit function theorem,
\begin{equation}
r'(p) = 1-\frac{g(r(p)-p)}{r(p)g^{\prime}(r(p)-p)}. \label{drr}
\end{equation}
By \autoref{ifp}, the optimal information provision policy depends on the location of \(\mu\): if \(\mu \in(0, r(p))\), \(\{0, r(p)\}\) is optimal; and if \(\mu \in[r(p), 1)\), \(\{\mu\}\) is optimal. Suppose \(\mu \in (0,r(p))\). Seller's payoff from setting \(p \in (1-a,\hat{p})\) is
\[
p \frac{\mu}{r(p)} G(r(p)-p) = p \mu g(r(p)-p),
\]
where the equality follows from \eqref{dfr}. Using \eqref{drr}, the derivative of this payoff is
\[
\mu g(r(p)-p)+p\mu g'(r(p)-p)\bigl(r'(p)-1\bigr)=\mu g(r(p)-p)\left(1-\frac{p}{r(p)}\right)>0,
\]
where the inequality follows from \(r(p)>p\). Thus, Seller's payoff is strictly increasing in \(p\) on \((1-a,\hat{p})\). For \(p \in [\hat{p},1)\), Seller's payoff is \(p \mu G(1-p)\). Hence, among the binary distributions identified above, the optimal one is full information, \(\{0,1\}\), and its optimal price is \(p_h\), the unique maximizer of \(pG(1-p)\).

Now suppose Seller uses no information \(\{\mu\}\); her payoff is given by \(pG(\mu - p)\), which is maximized at 
\begin{equation} \label{opm}
    p_\mu = \frac{G(\mu - p_\mu)}{g(\mu-p_\mu)}.
\end{equation}
Note that by definition of \(r(p)\), \(r(p) > p\) unless possibly at \(p=0\), which is never optimal. Then
\[
\frac{G(r(p_\mu) - p_\mu)}{g(r(p_\mu) - p_\mu)} = r(p_{\mu}) > p_\mu = \frac{G(\mu - p_\mu)}{g(\mu-p_\mu)},
\]
where the first equality follows from \eqref{dfr}, and the second equality holds by \eqref{opm}. By strict log-concavity of \(g\), it must be that \(\mu < r(p_\mu)\). Then the optimal distribution at \(p_\mu\) is \(\{0,r(p_\mu)\}\), which rules out providing no information as an optimal policy. To summarize, when \(p>1-a\), full information is optimal, and Seller's optimal price and profit are \(p_h\) and \(p_h\mu G(1-p_h)\), respectively.

Now suppose that \(p \le 1-a\). Again by the implicit function theorem,
\begin{align}
    t'(p) & = 1; \label{drt}
\end{align}
consequently, both \(r(p)\) and \(t(p)\) are increasing in \(p\), and by \eqref{drt}, 
\(t(p) = t(0) + p\).

If \(g(r(0)) \le 1/a\), it can be checked that for all \(p \le 1-a\), \(g(r(p)-p) \le 1/(p+a)\). By \autoref{ifp}, when \(\mu<p+a\), the distribution \(\{0,p+a\}\) is optimal; using this distribution, Buyer buys if and only if she buys without search, which happens with probability \(\mu/(p+a)\). Seller's expected payoff is therefore \(p\mu/(p+a)\), which is strictly increasing in \(p\). When \(\mu\ge p+a\), no information is optimal at the fixed price, and Seller's payoff is \(p\). In this case, \(p\le\mu-a<\mu(1-a)\), where \(\mu(1-a)\) is Seller's payoff from full information at \(p=1-a\). Hence, no price in the no-information branch can be globally optimal, while the payoff in the binary-distribution branch is maximized at \(p=1-a\), where the distribution is \(\{0,1\}\). Thus, among prices \(p\le1-a\), full information is optimal, with price \(1-a\) and expected payoff \(\mu(1-a)\).

Consider next the case in which \(g(r(0)) > 1/a\). Again by \autoref{ifp}, define
\[
\bar{p} = \sup\left\{p \in [0,1-a]: g(r(p)-p) > \frac{1}{p+a}\right\}.
\]
Strict log-concavity of \(g\) implies that \(g(r(p)-p)>1/(p+a)\) exactly when \(p<\bar p\), with equality when \(p=\bar p\). Now for a fixed \(p\), the optimal information provision depends on the location of \(\mu\). If \(\bar{p}\le p\le1-a\), \(\{0,p+a\}\) is optimal when \(\mu<p+a\), while \(\{\mu\}\) is optimal when \(\mu\ge p+a\). For \(p<\bar{p}\),
\begin{itemize}
    \item if \(\mu\in(0,r(p))\), \(\{0,r(p)\}\) is optimal,
    \item if \(\mu\in[r(p),t(p)]\), \(\{\mu\}\) is optimal,
    \item if \(\mu\in(t(p),p+a)\), \(\{t(p),p+a\}\) is optimal,
    \item if \(\mu\in[p+a,1)\), \(\{\mu\}\) is optimal.
\end{itemize}
No information cannot be globally optimal. When \(\mu\ge p+a\), Seller's payoff is \(p\le\mu-a<\mu(1-a)\), so full information at \(p=1-a\) is strictly better. When \(\mu<p+a\), Seller's payoff under no information is \(pG(\mu-p)\), and the argument surrounding \eqref{opm} applies unchanged. Moreover, \(p<\bar{p}\) implies that \(r(p)<p+a\); in the first bullet point, Seller's expected payoff is strictly increasing in \(p\), making it strictly better to charge \(\bar{p}\) and provide information according to \(\{0,\bar{p}+a\}\). Consequently, it only remains to consider the third bullet point.

To this end, suppose Seller provides information according to \(\{t(p), p+a\}\). Seller's problem of finding the optimal price is
\[
p\left[G(t(p)-p) \frac{p+a-\mu}{p+a-t(p)}+\frac{\mu-t(p)}{p+a-t(p)}\right];
\]
by \eqref{dft}, it can be written as
\[
p\left[G(t(0)) \frac{p+a-\mu}{a-t(0)}+\frac{\mu-p-t(0)}{a-t(0)}\right];
\]
because \(G(t(0)) < 1\), it is strictly concave in \(p\). Then the optimal price is given by 
\[
p_t = \frac{a G(t(0))-t(0)}{2[1-G(t(0))]}+\frac{\mu}{2}.
\]
For this price to be optimal, it must satisfy \(p_t<\bar p\). Observe that at \(\bar{p}\), \(r(\bar{p}) = t(\bar{p}) = t(0)+\bar{p}\). Then by definition of \(\bar{p}\),
\[
g(r(\bar{p})-\bar{p}) = \frac{G(r(\bar{p})-\bar{p})}{r(\bar{p})} = \frac{1}{\bar{p}+a},
\]
and hence
\[
\bar{p} = \frac{a G(t(0))-t(0)}{1-G(t(0))}. 
\]
But then \(p_t < \bar{p}\) implies that \(\mu < t(0) + p_t = t(p_t)\), which in turn implies that \(\{t(p_t),p_t+a\}\) is not optimal at \(p_t\). Therefore, the only remaining candidate among prices \(p\le1-a\) is \((p,\{0,p+a\})\) on its feasible range \(\mu<p+a\). Consequently, as in the case \(g(r(0))\leq1/a\), full information is optimal, \(p=1-a\) is the optimal price, and Seller's expected payoff is \(\mu(1-a)\). Summarizing, among prices \(p\le1-a\), full information at \(p=1-a\) is optimal. 

To conclude, full information is always optimal; the choice of the optimal selling strategy boils down to comparing \(1-a\) and \(p_hG(1-p_h)\). Then \(p = 1-a\) is optimal if and only if \(p_h G(1-p_h) \le 1-a\), and otherwise \(p_h\) is optimal. This yields the statement in the proposition. 
\end{proof}
\end{appendices}

\end{document}